\documentclass[pra,
twocolumn,
superscriptaddress,nofootinbib,longbibliography]{revtex4-2}

\usepackage[pdftex]{graphicx}
\usepackage{amsmath,amsfonts,amsbsy,amssymb,amsthm,mathtools}
\usepackage{mathrsfs}
\usepackage{epstopdf}
\usepackage{ulem,color}
\usepackage{enumitem}
\usepackage{siunitx}
\usepackage{amssymb,amsmath}
\usepackage{ulem,color}
\usepackage{MnSymbol}
\usepackage{environ}
\usepackage{xcolor}

\usepackage{tikz}
\usetikzlibrary{quantikz}
\usetikzlibrary{automata,arrows,positioning,calc}

\usepackage{pgfplots}
\usepackage{pstool}
\usepackage{tikzsymbols}
\usetikzlibrary{plotmarks,spy}
\usepgfplotslibrary{external} %
\newtheorem{theorem}{Theorem}
\newtheorem{proposition}[theorem]{Proposition}
\newtheorem{lemma}[theorem]{Lemma}

\theoremstyle{definition}
\newtheorem{definition}[theorem]{Definition}

\newtheorem{remark}[theorem]{Remark}
\newtheorem{example}[theorem]{Example}

\usepackage[colorlinks = true, citecolor=blue, linkcolor=blue, urlcolor=blue]{hyperref}
\usepackage[T1]{fontenc}
\usepackage{mwe}
\usepackage{tikz}
\usetikzlibrary{quantikz}
\newcommand{\ketbra}[1]{|{#1}\>\mkern-4mu\<{#1}|}
\newcommand{\tr}{\textup{Tr}}
\renewcommand{\>}{\rangle}
\newcommand{\<}{\langle}
\newcommand{\N}{{\mathbb{N}}} %
\newcommand{\E}{{\mathbb{E}}}

\newcommand{\negl}{{\mathrm{negl}}}
\newcommand{\poly}{{\mathrm{poly}\,}}
\newcommand{\rand}{{\mathtt{r}}}
\newcommand{\algo}{{\mathcal{A}}}

\newcommand{\CQT}{Centre for Quantum Technologies, National University of Singapore, 3 Science Drive 2, Singapore 117543.\looseness=-1}
\newcommand{\NTU}{Nanyang Quantum Hub, School of Physical and Mathematical Sciences, Nanyang Technological University, Singapore 639673.\looseness=-1}
\newcommand{\ihpc}{Institute of High Performance Computing (IHPC), Agency for Science, Technology and Research (A*STAR), 1 Fusionopolis Way, $\#$16-16 Connexis, Singapore 138632, Republic of Singapore.}
\newcommand{\qinc}{
Quantum Innovation Centre (Q.InC), Agency for Science Technology and Research (A*STAR), 2 Fusionopolis Way, Innovis $\#$08-03, Singapore 138634, Republic of Singapore. }

\begin{document}

\normalem
\newlength\figHeight 
\newlength\figWidth 

\title{Near-Term Pseudorandom and Pseudoresource Quantum States}

\author{Andrew Tanggara}
\email{andrew.tanggara@gmail.com}
\affiliation{\CQT}
\affiliation{\NTU}

\author{Mile Gu}
\email{mgu@quantumcomplexity.org}
\affiliation{\NTU}
\affiliation{\CQT}

\author{Kishor Bharti}
\email{kishor.bharti1@gmail.com}
\affiliation{\qinc}
\affiliation{\ihpc}

\date{\today}

\begin{abstract}
A pseudorandom quantum state (PRS) is an ensemble of quantum states indistinguishable from Haar-random states to observers with efficient quantum computers.
It allows one to substitute the costly Haar-random state with efficiently preparable PRS as a resource for cryptographic protocols, while also finding applications in quantum learning theory, black hole physics, many-body thermalization, quantum foundations, and quantum chaos.
All existing constructions of PRS equate the notion of efficiency to quantum computers which runtime is bounded by a polynomial in its input size.
In this work, we relax the notion of efficiency for PRS with respect to observers with near-term quantum computers implementing algorithms with runtime that scales slower than polynomial-time.
We introduce the \textit{$\mathbf{T}$-PRS} which is indistinguishable to quantum algorithms with runtime $\mathbf{T}(n)$ that grows slower than polynomials in the input size $n$.
We give a set of reasonable conditions that a $\mathbf{T}$-PRS must satisfy and give two constructions by using quantum-secure pseudorandom functions and pseudorandom functions.
For $\mathbf{T}(n)$ being linearithmic, linear, polylogarithmic, and logarithmic function, we characterize the amount of quantum resources a $\mathbf{T}$-PRS must possess, particularly on its coherence, entanglement, and magic.
Our quantum resource characterization applies generally to any two state ensembles that are indistinguishable to observers with computational power $\mathbf{T}(n)$, giving a general necessary condition of whether a low-resource ensemble can mimic a high-resource ensemble, forming a \textit{$\mathbf{T}$-pseudoresource} pair.
We demonstate how the necessary amount of resource decreases as the observer's computational power is more restricted, giving a $\mathbf{T}$-pseudoresource pair with larger resource gap for more computationally limited observers.
\end{abstract}

\maketitle

True randomness is a costly resource that lies at the foundation of many information processing tasks, including probabilistic computation and cryptography.
However to an observer with limited computational resource, one may design an object that looks random to this observer, mimicking a truly random object.
In quantum information processing, the Haar-random state is a truly random ensemble of quantum states that requires exponential time to generate.
A pseudorandom quantum state (PRS)~\cite{ji2018pseudorandom}, on the other hand, is an ensemble of quantum states which can be efficiently generated, but is indistinguishable from Haar-random quantum states by any efficient quantum algorithms up to a negligible probability, even given multiple copies of them (see Fig.~\ref{fig:indistinguishability}).
Since its inception in~\cite{ji2018pseudorandom}, many other constructions of PRS and its variants has been proposed~\cite{brakerski2019pseudo,brakerski2020scalable,aaronson2022quantum,lu2024quantum,giurgica2023pseudorandomness,ananth2022cryptography,ananth2022pseudorandom,bansal2024pseudorandom,morimae2025quantum,grilo2025quantumpseudoresourcesimplycryptography} and has direct application as cryptographic primitives~\cite{morimae2022quantum,ananth2022cryptography,morimae2025quantum,grilo2025quantumpseudoresourcesimplycryptography}, as well as applications in quantum learning theory~\cite{huang2022quantum}, black hole physics~\cite{bouland2019computational,engelhardt2025cryptographic,yang2025complexity}, many-body thermalization~\cite{feng2025dynamics}, and quantum chaos~\cite{gu2024simulating}.
On the other hand, its connections to quantum foundations such as entanglement~\cite{aaronson2022quantum,goulao2024pseudo,grewal2024pseudoentanglement,cheng2024pseudoentanglement}, coherence~\cite{haug2023pseudorandom}, and magic~\cite{Gu_2024} are also intriguing, particularly on how it can mimic high-resource states while actually possessing only a low amount of them, acting as a \textit{pseudoresource}~\cite{haug2023pseudorandom,bansal2024pseudorandom}.

Existing results on PRS equates the notion of computational efficiency for its indistinguishability to quantum algorithms running at most in polynomial-time in the number of qubits $n$ of the PRS.
Such pseudorandomness in the classical regime over bitstrings with respect to polynomial-time algorithms is widely applicable, as large-scale classical computers that can run them are widely available.
However, quantum computers are much more restrictive today where implementation of quantum algorithms only available for small instances, thus limiting the use of PRS.
With this problem in mind, we raise the questions of: 
How do one construct a PRS which is indistinguishable to small-scale quantum computers?
What are the properties of such PRS constructions computationally?
What quantum properties do these PRS have?
Do these relaxed PRS constructions require lesser resource?
Can they mimic entanglement, magic, and coherence using lesser resource than polynomial-time PRS?

In this work, we address these questions by proposing a framework that relaxes the polynomial-time computational indistinguishability of the usual notion of PRS to indistinguishability for observers with more restrictive computational resource.
We define the \textit{$\mathbf{T}$-PRS}, an ensemble of states indistinguishable from Haar-random states to quantum algorithms which runtime is bounded by a function that belongs to a family $\mathbf{T}$ which scales slower than polynomials.
As in the usual polynomial-time PRS, the indistinguishability property of $\mathbf{T}$-PRS holds up to a negligibly small probability, even when multiple copies are given to the distinguisher algorithm.
We characterize the negligibly small probability and how many copies of states can the algorithm process such that it cannot arbitrarily increase its probability of distinguishing the $\mathbf{T}$-PRS from Haar-random states, given its $\mathbf{T}$-bounded runtime.
Using these characterizations, we give two explicit constructions of $\mathbf{T}$-PRS by using quantum-secure pseudorandom permutations and quantum-secure pseudorandom functions, inspired by constructions in~\cite{aaronson2022quantum,giurgica2023pseudorandomness}.
For these constructions we consider $\mathbf{T}$ as a function $f(n)$ and polynomials of a function $\poly f(n)$, where $n$ is the number of qubits of the PRS.

We then analyze pair of quantum state ensembles indistinguishable to $\mathbf{T}$-bounded observers, one possessing high-resource and the other low-resource, which we call a \textit{$\mathbf{T}$-pseudoresource} pair.
For observers with quantum algorithms which runtime is bounded by function $\mathbf{T}(n)$ given by linearithmic ($O(n\log n)$), linear ($O(n)$), polylogarithmic ($O(\poly\log n)$), and logarithmic ($O(\log n)$) functions, we show that the necessary amount of resource (entanglement, coherence, and magic) in the low-resource ensemble decreases with $\mathbf{T}(n)$.
Since the $\mathbf{T}$-PRS are indistinguisable from Haar-random ensemble to size-$\mathbf{T}$ circuits, they are able to mimic high amount of entanglement, coherence, and magic of the Haar-random ensemble, with smaller amount of these resources compared to previous constructions of PRS.
We show the pseudoresource gaps between $\mathbf{T}$-PRS and Haar-random ensemble for different $\mathbf{T}$.
Compared to the recently proposed pseudorandom density matrices (PRDM)~\cite{bansal2024pseudorandom} which mimic high amount of entanglement, coherence, and magic with zero amount of these resources, the gap between perceived and actual resource of $\mathbf{T}$-PRS lies in between that of PRDM and PRS.

Below we give an outline of this paper.
In Section~\ref{sec:pseudorandomness_and_indistinguishability}, we lay out the framework to define the notion of pseudorandomness and indistinguishability with respect to observers with limited computational resource characterized by class of function $\mathbf{T}$.
Particularly, we discuss how the negligible probabilities with respect to $\mathbf{T}$ can be defined in Section~\ref{sec:negligible_distinguishability} to define the notion of computational indistinguishability wiht respect to $\mathbf{T}$, and finally $\mathbf{T}$-PRS in Section~\ref{sec:T-pseudorandomness}.
In Section~\ref{sec:TPRS_constructions}, we give two constructions of $\mathbf{T}$-PRS inspired by the subset phase state~\cite{aaronson2022quantum} and subset state~\cite{giurgica2023pseudorandomness}.
In Section~\ref{sec:T-pseudoresources}, we discuss pseudoresource state ensembles with respect to the observer's computational power characterized by $\mathbf{T}$.
Here we give a lower bound on the expected amount of resource of the low-resource ensemble and an upper bound on resource gap between the high and low-resource ensembles for coherence (Section~\ref{sec:coherence}), entanglement (Section~\ref{sec:entanglement}), and magic (Section~\ref{sec:magic}).

\section{Computational Pseudorandomness and Indistinguishability}\label{sec:pseudorandomness_and_indistinguishability}

\begin{figure}
    \centering
    \includegraphics[width=0.8\columnwidth]{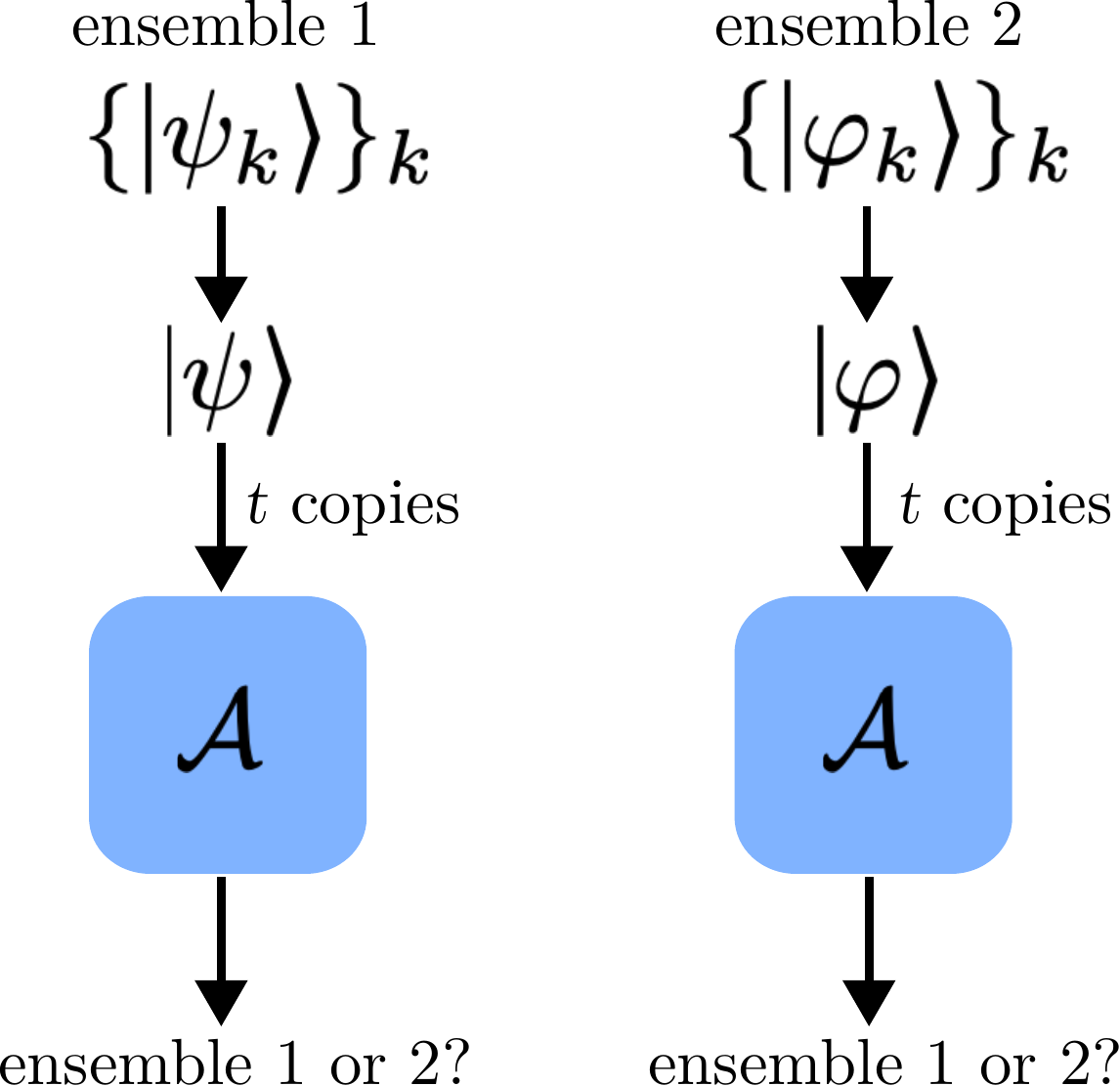}
    \caption{In this illustration, we consider $\mathbf{T}$-indistinguishable $n$-qubit pair of ensembles $\{|\psi_k\>\}_k$ and $\{|\varphi_k\>\}_k$ (as defined in Section~\ref{sec:T-pseudorandomness}).
    A quantum algorithm $\mathcal{A}$ is given input of either $t$ copies of $|\psi\>$ randomly sampled from $\{|\psi_k\>\}_k$ or $t$-copies of $|\varphi\>$ randomly sampled from $\{|\varphi_k\>\}_k$ such that it outputs $\mathcal{A}(\E_k[|\psi_k\>])\in\{0,1\}$ or $\mathcal{A}(\E_k[|\psi_k\>])\in\{0,1\}$ indicating which ensemble the input state belongs to.
    If the runtime of $\mathcal{A}$ given $t$ copies of $n$-qubit state input is given by $s(n)\in O(\mathbf{T}(n))$, then $\mathcal{A}$ cannot guess which ensemble the input belongs to, expect for a negligible probability (as defined in Section~\ref{sec:negligible_distinguishability}).
    }
    \label{fig:indistinguishability}
\end{figure}

Randomness is widely associated with the degree of uniformity in the frequency of each possible output from a particular source.
A source with perfect uniformity in the frequency of its outputs, therefore is a perfectly random source.
Pseudorandomness, on the other hand, is a source which is not perfectly random but is indistinguishable from a perfect random source.
Statistically, this can be quantified by ``how far'' the distribution of a source is from the uniform distribution or by how different the characters of these two distributions are.
However these statistical measures do not take into account the \textit{computational cost} of distinguishing such distributions.

In quantum systems, the objects that one concerns with are quantum states.
A perfect randomness can therefore be associated with an ensemble of quantum states which is distributed uniformly over all possible quantum states for a given system dimension.
This is captured by the Haar-random quantum state $\{|\varphi\>\}_\varphi$.
Thus a \textit{pseudorandom quantum state} (PRS) is a quantum state ensemble $\{|\psi\>\}_\psi$ which is indistinguishable from Haar-random state to observers with bounded computational resource.
The notion of indistinguishability holds up to some \textit{neglibible} probability even if the observers are given a bounded number of copies of states allowed by how much computational resource it has.

The usual $n$-qubit PRS is concerned with an observer with an access to a quantum computer to implement any quantum algorithm $\mathcal{A}$ which runtime is bounded by some polynomial in $n$.
Here for $t$-copies of $n$-qubit input state $|\tau\>^{\otimes t}$ to $\mathcal{A}$, it spits out output $\mathcal{A}(|\tau\>^{\otimes t})$ which is either 0 or 1, indicating whether it received PRS inputs ($\tau=\psi$) or Haar-random inputs ($\tau=\varphi$).
Given polynomially many $t(n)\in O(\poly n)$ copies of PRS $|\psi\>$ and Haar-random states $|\varphi\>$ over an $n$-qubit system, the probability of $\algo$ distinguishing $|\psi\>^{\otimes t(n)}$ and $|\varphi\>^{\otimes t(n)}$ is negligible:
\begin{equation}\label{eqn:polytime_indistinguishability}
    \Big|\Pr_\psi[\algo(|\psi\>^{\otimes t(n)})=1] - \Pr_\varphi[\algo(|\varphi\>^{\otimes t(n)})=1]\Big| < \eta(n)
\end{equation}
where $\eta$ is a negligible function, i.e. a function that scales slower than $\frac{1}{g(n)}$ for all polynomial $g\in\poly$.
For an illustration of this scenario see Fig.~\ref{fig:indistinguishability}.
We denote the set of all such function $\eta$ as $\negl_\poly$.
Note that in eqn.~\eqref{eqn:polytime_indistinguishability}, the input size to algorithm $\algo$ is $N=nt(n)$ (i.e. $t(n)$ copies of $n$-qubit states), which is polynomial in $n$.
Hence if runtime of $\algo$ is a polynomial $s(N)$ in $N$, then it runs for $s(nt(n))$ given input $|\psi\>^{\otimes t(n)}$ or $|\varphi\>^{\otimes t(n)}$, where $s(nt(n))$ is a polynomial in $n$ since composition of polynomials is a polynomial.
Note how here we distinguish between the number of qubits $n$ of \textit{individual} state $|\psi\>$ and the total number of qubits $N$ of \textit{multiple copies} of states $|\psi\>^{\otimes t}$.

For the rest of this section we will build a formal framework generalizing the polynomial-time indistinguishability discussed above.
In particular, we want to consider observers with different computational power by replacing algorithms with runtime bounded by a polynomial by those which runtime if bounded by some nondecreasing function $s:\N\rightarrow\N$ belonging a class of functions $\mathbf{T}$.
Namely, we want the runtime of $\mathcal{A}$ on $N$-qubit input to be bounded by $s(N)\in O(\mathbf{T}(N))$ \footnote{
We give some explanation for our notation.
For arbitrary class of functions $\mathbf{T}$ we sometimes write $\mathbf{T}(n)$ instead of $\mathbf{T}$ to emphasize the variable of the functions in $\mathbf{T}$, i.e. function $f(n)$ in $\mathbf{T}(n)$.
Also for some parts in the rest of the paper, we often write arithmetic operations over set of functions to indicate arithmetic operations over arbitrary function in these sets, which is a standard convention in writing asymptotics.
For example, we write $O(f(n)) + \negl_\poly = o(l(n))$ when we mean $g(n) + h(n) = q(n)$ for some $g\in O(f(n))$, $h\in \negl_\poly$, and $q\in o(l(n))$.
In some parts, we write the latter to make clearer arguments, however in parts where the context is clear we write in the former.
We also use this notation on arbitrary class of functions $\mathbf{T}$ (e.g. $O(\mathbf{T}) + \negl_\mathbf{T}$).
Note also that sometimes $\mathbf{T}$ may indicate a single function, e.g. $\mathbf{T}=\log n$, as opposed to a family of function as in the case of $\mathbf{T}=\poly n$ where $\mathbf{T}$ is the set of all polynomials in $n$.
}.
Hence the number of copies $t(n)$ of $n$-qubit states $|\tau\>$ must satisfy $s(N)=s(nt(n)) \in O(\mathbf{T}(n))$.
With such observers, then we need to formulate a different notion of negligibility that still impose the restriction that any such observer with computational resource bounded by $\mathbf{T}$ cannot arbitrarily increase the probability of distinguishing $n$-qubit states $|\psi\>$ and $|\varphi\>$ given $t(n)$ copies of them.
Therefore in summary, for a class of function $\mathbf{T}$ we want to impose these requirements on the observer's computational power, negligibility, and number of copies with respect to the number of qubits $n$ of an individual copy of states in question:
\begin{enumerate}

    \item Given any number of copies $t(n)$ of any $n$-qubit state $|\tau\>$ as an input to any algorithm $\algo$ chosen by the observer, the runtime of computing its output $\algo(|\tau\>^{\otimes t(n)})$ is bounded by $s(n) \in O(\mathbf{T}(n))$.

    \item For the observer with computational resource bounded by $\mathbf{T}$, the negligible probability $\eta(n)$ of its chosen algorithm $\mathcal{A}$ distinguishing $|\psi\>$ and $|\varphi\>$ given $t(n)$ copies of them is preserved even when composing $\mathcal{A}$ with other algorithm $\mathcal{A}'$ it has access to or by running $\mathcal{A}$ repeatedly with total runtime still bounded as $O(\mathbf{T}(n))$.
\end{enumerate}
After we have these requirements characterized, we then introduce the $\mathbf{T}$-PRS, quantum state ensembles which are indistinguishable from Haar-random states to observers with access to quantum algorithms which runtime is bounded by $s(n)\in O(\mathbf{T})$.

\subsection{Negligible distinguishability}\label{sec:negligible_distinguishability}

For polynomial-time observers, the corresponding set of negligible functions $\negl_\poly$ are chosen such that when the observer repeat the experiment polynomial number of times (since it has access to polynomial-depth algorithms), it cannot arbitrarily increase the probability of distinguishing $|\psi\>^{\otimes t(n)}$ and $|\varphi\>^{\otimes t(n)}$.
This notion, which motivates the polynomial-time indistinguishability of eqn.~\eqref{eqn:polytime_indistinguishability}, consists of two components: (1) a set of functions $\mathbf{N}$ which signifies negligible probability of distinguishability and (2) another set of \textit{repeat functions} $\mathbf{R}$ which signifies how many repetition of the experiment is allowed.
We will come back later to which set of repeat function $\mathbf{R}$ is allowed for different observers.
Formally, the aforementioned two criteria for neglibility can be stated as the following \textit{closure properties}.

\begin{definition}[{Closure properties}]\label{def:closure_property}
    Consider a pair of sets $\mathbf{N}$ and $\mathbf{R}$ of non-decreasing functions $g:\N\rightarrow\N$.
    We say that $\mathbf{N}$ satisfy the \textit{closure properties} with respect to \textit{repeat functions} $\mathbf{R}$ if for all $\eta_1,\eta_2\in\mathbf{N}$, it holds that:
    \begin{enumerate}
        \item $\eta_1(n) + \eta_2(n) \in\mathbf{N}$, and
        \item $r(n) \eta_1(n) \in\mathbf{N}$ for any $r\in\mathbf{R}$.
    \end{enumerate}
\end{definition}

\begin{remark}
    The first closure property concerns two algorithms $\algo$ and $\algo'$ with probabilities of distinguishing PRS and Haar-random states bounded by $\eta_1$ and $\eta_2$, respectively (in the sense of eqn~\eqref{eqn:polytime_indistinguishability}).
    This property guarantees that the two algorithms combined together still give a negligible probability.
    Particularly if we denote the event $S$ as a successful distinction from algorithm $\algo$ and event $S'$ for algorithm $\algo'$, the probability of one or both of them being successful is
    \begin{equation}
        \Pr[S\vee S'] \leq \Pr[S]+\Pr[S'] < \eta_1(n)+\eta_2(n) \in\mathbf{N}
    \end{equation}
    where the first inequality is given by the union bound.
    
    On the other hand, the second property concerns an algorithm $\algo$ that is run repeatedly $r(n)$ number of times.
    This property guarantees that whenever the probability of successfully distinguishing PRS and Haar-random states $\Pr[S]$ is bounded above by $\eta_1$, repeating it $r(n)\in\mathbf{R}$ many times still give a negligible probability.
    More precisely by using the union bound, the probability of some repetition of experiment being successful is
    \begin{equation}
        \Pr[ S_1 \vee\dots\vee S_{r(n)}] \leq \sum_{i=1}^{r(n)} \Pr[S_i] < r(n)\eta_1(n) \in\mathbf{N}
    \end{equation}
    where $S_i$ indicates the event of successful distinction in the $i$-th repetition of the experiment.
\end{remark}

In classical cryptography, negligible functions $\mathbf{N}=\negl_\poly$ with respect to polynomial-time observers has been shown in~\cite[Proposition 3.6]{katz2007introduction} to satisfy the closure properties with respect to $\mathbf{R}=O(\poly(n))$.
Here a polynomial-time observer may repeat the experiment any polynomial number of times, since product between any two polynomials is a polynomial.
Therefore, the closure properties prohibit any such observers from arbitrarily increasing the probability of distinguishing PRS from Haar-random states.
For more discussion on indistinguishability in classical pseudorandomness, see~\cite[Chapter 8.8]{katz2007introduction}.

Now we formalize the notion of negligibility with respect to a class of functions $\mathbf{T}$.

\begin{definition}[{$\mathbf{T}$-negligible functions}]\label{def:negligible_function}
    For a set of $\N\mapsto\N$ functions $\mathbf{T} \subseteq \{f:\N\rightarrow\N\}$, a function $\eta:\N\rightarrow[0,1]$ is $\mathbf{T}$-\textit{negligible} whenever for all $g\in\Theta(\mathbf{T})$ it holds that
    \begin{equation}
        \eta(n) < \frac{1}{g(n)}
    \end{equation}
    for all but finitely many $n\in\N$.
    The set of all $\mathbf{T}$-negligible functions is denoted as $\negl_\mathbf{T}$.
\end{definition}

Before we discuss how $\mathbf{T}$-negligible functions $\negl_\mathbf{T}$ plays a role in formulating indistinguishability for observers with different computational power, we will go through a few examples of negligible functions.

\begin{remark}[Polynomial-time negligible functions]\label{rem:poly_negl_function}
    For $\mathbf{T}=\poly n$, we show that we will recover the usual negligible function with respect to a polynomial-time algorithms.
    Namely a function $\eta\in\negl_\poly$ satisfies $\eta(n) < g(n)^{-1}$ for all $g\in\Theta(\poly n)$.
    Since for all such function $g$ there exists some $c>0$ and $N\in\N$ such that $n\geq N \Rightarrow g(n)\leq n^c$, therefore $\eta$ must satisfy
    \begin{equation}
        \eta(n) < \frac{1}{n^c} \;.
    \end{equation}
    for all $c>0$ and all but finitely many $n$.
\end{remark}

\begin{example}[Linearithmic-time negligible functions]
    For $\mathbf{T}=n\log n$, a function $\eta\in\negl_{n\log n}$ satisfies $\eta(n) < g(n)^{-1}$ for all $g\in \Theta(n\log n)$.
    In this case, for all such function $g$ there exists some $c>0$ and $N\in\N$ such that $n\geq N \Rightarrow g(n)\leq c n\log n$.
    Thus, $\eta$ must satisfy
    \begin{equation}
        \eta(n) < \frac{1}{c n\log n}
    \end{equation}
    for all $c>0$ and for all but finitely many $n$.
\end{example}

\begin{example}[Polylog-time negligible functions]
    For $\mathbf{T}=\poly\log n$, a function $\eta\in\negl_{\poly\log n}$ satisfies $\eta(n) < g(n)^{-1}$ for all $g\in \Theta(\poly\log n)$.
    Thus for all $c>0$, a $\poly\log n$-negligible function $\eta$ must satisfy
    \begin{equation}
        \eta(n) < \frac{1}{(\log n)^c}
    \end{equation}
    for all but finitely many $n$.
\end{example}

Note that for polynomial-time negligible functions, we can set the set of repeat functions $\mathbf{R}$ as the set of polynomially bounded functions, i.e. $\mathbf{R}=O(\poly n)$.
This makes sense since the observer is restricted to algorithms $\algo(|\tau\>)$ which runtime is bounded by some polynomial in $n$ (for some arbitrary state $|\tau\>$).
Now for observers with access to algorithms with runtime bounded by function $s\in O(\mathbf{T})$ this is not necessarily true.
For example, for $O(\mathbf{T})=O(n)$ we have an algorithm $\mathcal{A}(|\tau\>)$ that runs in time $s(n)\in O(n)$.
If we set $\mathbf{R} = O(n)$ and construct an algorithm $\algo'$ which repeats $\mathcal{A}(|\tau\>)$ by $r(n)=n$ times (hence taking input of $r(n)$ copies of $|\tau\>$), then $\algo'$ has a total runtime of $r(n)s(n) = ns(n) \in O(n^2)$.
This is not allowed since the observer is restricted only to algorithms that runs in $O(n)$.
Thus we want $\mathbf{R}$ to be the set of functions that signify the most number of repetition of any $O(\mathbf{T})$ time algorithm that the observer can do.
We now formulate this additional criteria.

\begin{definition}[{Repetition consistency}]\label{def:repetition_consistency}
    The set of repeat functions $\mathbf{R}$ is \textit{consistent} with $\mathbf{T}$ if for any $s\in O(\mathbf{T})$ we have $r(n)s(n)\in O(\mathbf{T})$ for all $r\in\mathbf{R}$.
\end{definition}

Now we show that $\mathbf{T}$-negligible functions $\negl_\mathbf{T}$ for $O(\mathbf{T})=O(f(n))$ and for $O(\mathbf{T})=O(\poly f(n))$ does have nice properties with respect to some corresponding repeat functions.
Namely they satisfy closure properties criteria in Definition~\ref{def:closure_property} and repetition consistency criteria in Definition~\ref{def:repetition_consistency}.

\begin{proposition}\label{prop:closure_negligible_function}
    The set of $f(n)$-negligible functions $\negl_{f(n)}$ for any non-decreasing, non-constant function $f:\N\rightarrow\N$ satisfy the closure properties with respect to the set of repeat functions $\mathbf{K}$, where $\mathbf{K}$ is the set of constant functions.
    Moreover $\mathbf{K}$ is consistent with $f(n)$.
\end{proposition}
\begin{proof}
    By the definition of $f(n)$-neglibible function, it holds for $i\in\{1,2\}$ that $\eta_i(n) < \frac{1}{g(n)}$ for all $g(n)\in \Theta(f(n))$.
    Thus for any $c_1,c_2>0$, there exists $N_1,N_2\in\N$ such that $n\geq N_i \Rightarrow \eta_i(n) < c_i/f(n)$.
    We set $N = \max\{N_1,N_2\}$ for each pair of $c_1,c_2$ so that
    \begin{equation}
        \eta_1(n) + \eta_2(n) < \frac{c_1}{f(n)} + \frac{c_2}{f(n)} \;.
    \end{equation}
    Hence for any $c=c_1+c_2>0$ there exists $N\in\N$ such that $n\geq N \Rightarrow \eta_1(n) + \eta_2(n) < c/f(n)$.
    Thus we have shown the first closure property $\eta_1(n) + \eta_2(n) \in\negl_{f(n)}$.

    To show the second closure property, again note that for any $c_1>0$ there exists $N\in\N$ such that $n\geq N \Rightarrow \eta_1(n) < \frac{c_1}{f(n)}$.
    Thus for a constant function $r(n)=c$ for some $c>0$ and for any $c_1$ we have
    \begin{equation}
        r(n)\eta_1(n) < c\frac{c_1}{f(n)}
    \end{equation}
    for all but finitely many $n$.
    Thus $r(n)\eta_1(n)\in \negl_{f(n)}$.

    Lastly to show that $\mathbf{K}$ is consistent with $f(n)$, simply note that for a constant function $r\in\mathbf{K}$ we have $r(n)=c$ for some $c>0$.
    Thus for any $s\in O(f(n))$ we have $r(n)s(n) = cs(n)$, which is in $O(f(n))$.
\end{proof}

\begin{proposition}
    The set of $\poly f(n)$-negligible functions $\negl_{\poly f(n)}$ for any non-decreasing, non-constant function $f:\N\rightarrow\N$ satisfy the closure properties with respect to repeat functions $O(\poly f(n))$. 
    Moreover $O(\poly f(n))$ is consistent with $\poly f(n)$.
\end{proposition}
\begin{proof}
    By the definition of $\poly f(n)$-neglibible function, it holds for $i\in\{1,2\}$ that $\eta_i(n) < \frac{1}{g(n)}$ for all $g(n)\in \Theta(\poly f(n))$.
    Thus for any $c_1,c_2>0$, there exists $N_1,N_2\in\N$ such that $n\geq N_i \Rightarrow \eta_i(n) < \frac{1}{f(n)^{c_i}}$.
    We set $N = \max\{N_1,N_2\}$ for each pair of $c_1,c_2$ so that
    \begin{equation}
        \eta_1(n) + \eta_2(n) < \frac{1}{f(n)^{c_1}} + \frac{1}{f(n)^{c_2}} \;.
    \end{equation}
    Hence $\eta_1(n) + \eta_2(n)$ is bounded by the inverse of some polynomial in $f(n)$ for all but finitely many $n$.
    Thus we have shown the first closure property $\eta_1(n) + \eta_2(n) \in\negl_{\poly f(n)}$.

    To show the second closure property, again note that for any $c_1>0$ there exists $N\in\N$ such that $n\geq N \Rightarrow \eta_1(n) < \frac{1}{f(n)^{c_1}}$.
    Now note that for a function $r(n)\in O(\poly f(n))$, there exists some $c>0$ and $N\in\N$ such that $n\geq N \Rightarrow r(n) < f(n)^c$.
    So we can pick any $c_1$ larger than $c$ so that we can obtain
    \begin{equation}
        r(n)\eta_1(n) < \frac{1}{f(n)^{c'}}
    \end{equation}
    for all $c'=c_1-c$ and for all but finitely many $n$.
    Thus $r(n)\eta_1(n)\in \negl_{\poly f(n)}$.

    Lastly to show that $O(\poly f(n))$ is consistent with $\poly f(n)$, simply note that for a function $g(n)\in O(\poly f(n))$ there exists some $c>0$ such that $g(n)< f(n)^c$ for all but finitely many $n$.
    Thus for any $r,s\in O(\poly f(n))$ there exists some $c'>0$ such that $r(n)s(n) < f(n)^{c'}$ for all but finitely many $n$, which shows that $r(n)s(n)$ is in $O(\poly f(n))$.
\end{proof}

\subsection{T-Pseudorandom Quantum States}\label{sec:T-pseudorandomness}

Before we define PRS with respect to observers with different computational resource, recall that in the usual polynomial-time PRS, the polynomial-time algorithm $\mathcal{A}$ used by the observer may receive an input of at most polynomially many $t(n)\in O(\poly n)$ copies of $n$-qubit state $|\tau\>$.
As we have discussed the runtime of $\mathcal{A}(|\tau\>^{\otimes t(n)})$ is still bounded by some polynomial, since composition of polynomials is itself a polynomial.
When the observer is restricted to algorithms that runs in $O(\mathbf{T})$, it is not the case in general that the runtime of $\mathcal{A}(|\tau\>^{\otimes t(n)})$ is still bounded by some function $s(n)\in O(\mathbf{T})$ if $t(n)\in O(\mathbf{T})$.
For example if the observer has access to algorithms $\mathcal{A}$ which runtime bounded by some function $s(N)\in O(N)$, where $N$ is the number of input qubits to $\mathcal{A}$ and $t(n)\in O(n)$, then given $t(n)$ copies of $n$-qubit state $|\tau\>$, $\mathcal{A}(|\tau\>^{\otimes t(n)})$ runtime is bounded by $s(nt(n)) \in O(n^2)$.
This is not allowed since we require that the observer only has access to quantum algorithms with runtime bounded by some function in $O(n)$.
To remedy this, we restrict the number of copies $t(n)$ of $n$-qubit state $|\tau\>$ such that $s(nt(n))\in O(\mathbf{T}(n))$ for any $s(N)\in O(\mathbf{T}(N))$.

Putting together this criterion on the number of copies with the criteria for $\mathbf{T}$-negligible functions, we can now formally define what it means for two ensembles to be indistinguishable with respect to $\mathbf{T}$.

\begin{definition}[{$\mathbf{T}$-indistinguishability}]\label{def:T_indistinguishable}
    Two ensembles of $n$-qubit states $\{|\psi\>\}_\psi$ and $\{|\varphi\>\}_\varphi$ are \textit{$\mathbf{T}$-indistinguishable} whenever for any quantum algorithm $\algo$ with $N$-qubit input outputting either 0 or 1 with runtime bounded by function $s(N)\in O(\mathbf{T}(N))$ and for all function $t(n)$ such that $s(nt(n))\in O(\mathbf{T}(n))$, it holds that
    \begin{equation}
        \Big|\Pr_\psi[\algo(|\psi\>^{\otimes t(n)})=1] - \Pr_\varphi[\algo(|\varphi\>^{\otimes t(n)})=1]\Big| < \eta(n) 
    \end{equation}
    for some $\mathbf{T}$-negligible function $\eta\in\negl_\mathbf{T}$.
\end{definition}

Now we give the definition of a $\mathbf{T}$-PRS: a PRS which are indistinguishable from Haar-random states to observers with an access to $O(\mathbf{T})$-time algorithms.

\begin{definition}[{$\mathbf{T}$-pseudorandom quantum states ($\mathbf{T}$-PRS)}]\label{def:T_pseudorandom_states}
    Consider a set of $\N\mapsto\N$ functions $\mathbf{T} \subseteq \{f:\N\rightarrow\N\}$.
    For $n\in\N$, an ensemble of $n$-qubit states $\{|\psi_k\> : k\in \mathcal{K}_n\}$ over \textit{keyspace} $\mathcal{K}_n$ with $|\mathcal{K}_n|=l(n)\in O(\mathbf{T}(n))$ is a $\mathbf{T}$-\textit{pseudorandom state} ($\mathbf{T}$-PRS) if it satisfies:
    \begin{enumerate}
        \item There exists a uniform quantum circuit $\{G_n\}_n$ with size $g(n)\in O(\poly n)$ that outputs an $n$-qubit quantum state $G_n(k)=|\psi_k\>$ given input $k$.

        \item Ensemble $\{|\psi\>\}_\psi$ and $n$-qubit Haar-random state ensemble $\{|\varphi\>\}_\varphi$ are $\mathbf{T}$-indistinguishable as defined in Definition~\ref{def:T_indistinguishable}.

        \item The set of negligible functions $\negl_\mathbf{T}$ must satisfy the closure properties with respect to some repeat function $\mathbf{R}$ consistent with $\mathbf{T}$ as defined in Definition~\ref{def:closure_property} and Definition~\ref{def:repetition_consistency}.
    \end{enumerate}
\end{definition}

Note that here the bound for the $\mathbf{T}$-PRS generator is the same as the polynomial-time PRS, namely that we demand the generator must be a polynomial-size circuit regardless of $\mathbf{T}$ bound on the computational resource of the observer.
This can be thought of as a scenario where the generator belongs to a party with more computational resource than the observer, which is the focus of this work.
A more general scenario where the computational resource of the generator is also bounded by $\mathbf{T}$ for \textit{any} choice of $\mathbf{T}$ is left as an open question for future work.

\section{T-Pseudorandom Quantum State Constructions}\label{sec:TPRS_constructions}

In this section we give two different constructions of $\mathbf{T}$-PRS.
The first construction is inspired by the subset phase state construction proposed in~\cite{aaronson2022quantum}, whereas the second construction takes inspirations from the subset state proposed in~\cite{giurgica2023pseudorandomness}.
These constructions use quantum-secure pseudorandom phase functions (QPRPF) and quantum-secure pseudorandom permutations (QPRP) as primitives.
As their name indicate, these functions (permutations) are efficiently computable functions (permutations) that are indistinguihable from truly random functions (permutations) to efficient quantum algorithms.
Following what we have done so far in generalizing efficiency of quantum algorithm to $\mathbf{T}$-efficient, where its runtime is bounded by some function $s\in O(\mathbf{T})$, we first need the analogous notion of $\mathbf{T}$-QPRPF and $\mathbf{T}$-QPRP.

\begin{definition}[{Quantum-secure pseudorandom phase functions and quantum-secure pseudorandom permutations}]
    For keyspace $\mathcal{K}$ and $n\in\N$, a family of phase functions $F=\{f_k : \{0,1\}^n\rightarrow\{0,1\} \}_{k\in\mathcal{K}}$ is a \textit{$\mathbf{T}$ - quantum-secure pseudorandom phase function} ($\mathbf{T}$-QPRPF) if $f_k$ is computable in $O(\mathbf{T}(n))$ time and for all quantum algorithm $\mathcal{A}$ running in $O(\mathbf{T}(n))$ time, it holds that
    \begin{equation}
        \Big| \Pr_{k\leftarrow\mathcal{K}}[\mathcal{A}^{f_k}(1^n)=1] - \Pr_{\rand_f}[\mathcal{A}^{\rand_f}(1^n)=1] \Big| = \eta(n) \;.
    \end{equation}
    A family of permutations $=\{\sigma_k:\{0,1\}^n\rightarrow\{0,1\}^n\}_{k\in\mathcal{K}}$ is \textit{$\mathbf{T}$ - quantum-secure pseudorandom permutation} ($\mathbf{T}$-QPRP) if $\sigma_k$ is computable in $O(\mathbf{T}(n))$ time and for all quantum algorithm $\mathcal{A}$ running in $O(\mathbf{T}(n))$, it holds that
    \begin{equation}
        \Big| \Pr_{k\leftarrow\mathcal{K}}[\mathcal{A}^{\sigma_k}(1^n)=1] - \Pr_{\rand_\sigma}[\mathcal{A}^{\rand_\sigma}(1^n)=1] \Big| = \negl_\mathbf{T}(n) \;.
    \end{equation}
    Here, $\rand_f$ and $\rand_\sigma$ are uniformly-random phase function and uniformly-random permutation, respectively, and $\mathcal{A}^{\sigma_k},\mathcal{A}^{\rand_\sigma}$ denotes quantum algorithm $\mathcal{A}$ with oracle access to $\sigma_k,\sigma_k^{-1}$ and $\rand_\sigma,\rand_\sigma^{-1}$.
\end{definition}

By using $\mathbf{T}$-QPRPF and $\mathbf{T}$-QPRP we will now show the constructions of $\mathbf{T}$-pseudorandom subset phase states and $\mathbf{T}$-pseudorandom subset states.

\subsection{T-pseudorandom subset phase states}\label{sec:subset_phase_states}

\begin{definition}\label{def:subset_phase_state}
    For a subset of $n$-bit string $S\subseteq\{0,1\}^n$ and binary function $f:\{0,1\}^n\rightarrow\{0,1\}$, an $f,S$-\textit{subset phase state} is defined as
    \begin{equation}
        |\psi_{f,S}\> = \frac{1}{|S|} \sum_{x\in S} (-1)^{f(x)} |x\> \;.
    \end{equation}
    For permutation $\sigma:[n]\rightarrow[n]$ (where $[n]:=\{1,\dots,n\}$), an $f,\sigma$-subset phase state is defined as
    \begin{equation}
        |\psi_{f,\sigma}\> = \frac{1}{\sqrt{2^m}} \sum_{x\in\{0,1\}^m} (-1)^{f(p_\sigma(x0^{n-m}))} |p_\sigma(x0^{n-m})\> \;.
    \end{equation}
    where $p_\sigma:\{0,1\}^n\rightarrow\{0,1\}^n$ permutes the order of $n$-bit string $w$ as $p_\sigma(w) = w_{\sigma(1)}\dots w_{\sigma(n)}$.
\end{definition}

Now we will describe how one can construct a subset phase state that is a $\mathbf{T}$-PRS.
First we will describe the generator circuit of the $n$-qubit $f,S$-subset phase state.

\begin{lemma}[{\cite{aaronson2022quantum}}]\label{lem:subset_phase_state_generator}
    An $n$-qubit $f,S$-subset phase state with $|S|=2^m$ can be generated by a circuit with depth $O(\poly n)$.
\end{lemma}

This is shown in~\cite{aaronson2022quantum} by a construction of a circuit that takes $n$-qubit input and apply hadamard gates on the first $m$ qubits, then apply the permutation $\sigma$, and then the phase oracle $U_f$.

It is shown in Theorem 2 of~\cite{aaronson2022quantum} that the trace distance between an $n$-qubit truly random subset phase state and an $n$-qubit Haar-random ensemble $\{|\varphi\>\}$ is bounded as:
\begin{equation}\label{eqn:stat_indistinguishable_random_subsetphase_haar}
    d_{\tr}\Big( \E_{\rand_f,\rand_S}\big[\ketbra{\psi_{\rand_f,\rand_S}}^{\otimes t}\big] \,,\, \E_\varphi\big[\ketbra{\varphi}^{\otimes t}\big] \Big) < O\Big(\frac{t^2}{2^m}\Big)
\end{equation}
for $t<2^m<2^n$ where $\rand_f$ is uniformly-random over all phase functions $\rand_f:\{0,1\}^n\rightarrow\{0,1\}$ and $\rand_S$ is uniformly-random over all subsets of size $|S|=2^m$ and $\mu$ is the $n$-qubit Haar measure.
Uniformly-random subset phase state can be equivalently obtained by uniformly-random permutation $\rand_\sigma$ and uniformly-random phase function $\rand_f$,
\begin{equation}
    \E_{\rand_f,\rand_\sigma} \big[\ketbra{\psi_{\rand_f,\rand_\sigma}}^{\otimes t}\big] \;.
\end{equation}
where
\begin{equation}
    |\psi_{\rand_f,\rand_\sigma}\> = \frac{1}{\sqrt{2^m}} \sum_{x\in\{0,1\}^m} (-1)^{\rand_f(\rand_\sigma(x0^{n-m}))} |\rand_\sigma(x0^{n-m})\> \,.
\end{equation}
Now we will show how to determine the size of subset $S\subseteq\{0,1\}^n$ for the $\mathbf{T}$-PRS subset phase state construction.

\begin{proposition}\label{prop:subset_phase_state_subset_size}
    Let $f:\N\rightarrow\N$ be a non-decreasing function that grows at most polynomially.
    It holds that:
    \begin{enumerate}
        \item For number of copies $t:=t(n)\in O(1)$ and size of subset $|S|=2^m:=2^{m(n)} \in \omega(f(n))$, the trace distance in eqn.~\eqref{eqn:stat_indistinguishable_random_subsetphase_haar} is $f(n)$-negligible.
        \item For number of copies $t:=t(n)\in O(\poly f(n))$ and size of subset $|S|=2^m:=2^{m(n)} \in \omega(\poly f(n))$, the trace distance in eqn.~\eqref{eqn:stat_indistinguishable_random_subsetphase_haar} is $\poly f(n)$-negligible.
    \end{enumerate}
\end{proposition}
\begin{proof}
    For $O(\mathbf{T})=O(f(n))$, set the number of copies as $t(n)\in O(1)$ and size of subset as $2^{m(n)} \in \omega(f(n))$.
    Hence there exists $c>0$ and $N$ such that $n\geq N\Rightarrow t(n)\leq c$ and for all $c'>0$ there exists $N$ such that $n\geq N\Rightarrow 2^{m(n)} > c'f(n)$ (or equivalently, $m(n) > \log(c'f(n))$).
    Hence it holds that for all $c'>0$ and for some $c>0$,
    \begin{equation}
        \frac{t(n)^2}{2^{m(n)}} < \frac{c^2}{c'f(n)}
    \end{equation}
    for all but finitely many $n$, which implies that $\frac{t(n)^2}{2^{m(n)}} \in o(f(n)^{-1})$.
    Thus for any $g(n)\in O(\frac{t^2}{2^m})$ with $t:=t(n)\in O(1)$ and $2^m:=^{m(n)}\in \omega(f(n))$ we have $g(n)\in o(f(n)^{-1})$, and therefore $g(n)$ is a $\negl_{O(f(n))}$ function.
    So by eqn.~\eqref{eqn:stat_indistinguishable_random_subsetphase_haar} the trace distance between $t(n)\in O(1)$ copies of $n$-qubit subset phase state ensemble with $|S|=\omega(f(n))$ and $t(n)\in O(1)$ copies of $n$-qubit Haar-random ensemble is $\mathbf{T}$-negligible.

    Now consider $O(\mathbf{T}) = O(\poly f(n))$ and $t(n)\in O(\poly f(n))$ and subset size $|S|=2^{m(n)}\in \omega(\poly f(n))$.
    Thus for all $c'>0$ and for some $c>0$ it holds that
    \begin{equation}
        \frac{t(n)^2}{2^{m(n)}} < \frac{f(n)^{2c}}{f(n)^{c'}}
    \end{equation}
    for all but sufficiently many $n$.
    Since $c'>0$ can be arbitrarily large therefore for all $g(n)\in O(\frac{t(n)^2}{2^{m(n)}})$ it holds that $g(n)\in o(\frac{1}{\poly f(n)})$, which implies that $g(n) \in \negl_{\poly f(n)}$.
    Therefore by eqn.~\eqref{eqn:stat_indistinguishable_random_subsetphase_haar} the trace distance between $t(n)\in O(\poly f(n))$ copies of $n$-qubit subset phase state ensemble with $|S|=\omega(\poly f(n))$ and $t(n)\in O(\poly f(n))$ copies of $n$-qubit Haar-random ensemble is $\mathbf{T}$-negligible.
\end{proof}

\begin{remark}\label{rem:number_of_copies}
    Note that the reason that we consider $t(n)\in O(1)$ and $t(n)\in O(\poly f(n))$ is to satisfy the $f(n)$-indistinguishability and $\poly f(n)$-indistinguishability, respectively.
    Particularly by the definition of $\mathbf{T}$-indistinguishability in Definition~\ref{def:T_indistinguishable}, we need an algorithm $\mathcal{A}$ with runtime $s(N)\in O(\mathbf{T}(N))$ given $N$-qubit input to run in $s(nt(n))$ time where $s(nt(n))\in O(\mathbf{T}(n))$ given $t(n)$ copies of an $n$-qubit state $|\tau\>$.
    For $\mathbf{T}(n)=f(n)$, $s(nt(n))\in O(f(n))$ is satisfied when $t(n)=O(1)$.
    On the other hand for $\mathbf{T}(n)=\poly f(n)$, $s(nt(n))\in O(\poly f(n))$ is satisfied when $t(n)=O(\poly f(n))$.
\end{remark}

Finally, by using Lemma~\ref{lem:subset_phase_state_generator} and Proposition~\ref{prop:subset_phase_state_subset_size} we obtain a subset phase state $\mathbf{T}$-PRS construction.

\begin{theorem}\label{thm:subset_phase_state_Tpseudorandomness}
    Let $f:\N\rightarrow\N$ be a non-decreasing function that grows at most polynomially.
    It holds that:
    \begin{enumerate}
        \item A subset phase state ensemble $\{|\psi_{f,\sigma}\>\}_{f,\sigma}$ with 
        subset size $|S|\in \omega(f(n))$ is a $f(n)$-PRS given number of copies $t:=t(n)\in O(1)$.
        
        \item A subset phase state ensemble $\{|\psi_{f,\sigma}\>\}_{f,\sigma}$ with 
        subset size $|S| \in \omega(\poly f(n))$ is a $\poly f(n)$-PRS given number of copies $t:=t(n)\in O(\poly f(n))$.
    \end{enumerate}
\end{theorem}
\begin{proof}
    These subset phase states can be generated in $O(n)$ by Lemma~\ref{lem:subset_phase_state_generator}, so we only need to show that it is $\mathbf{T}$-indistinguishable to Haar-random state ensembles for negligible functions $\negl_\mathbf{T}$ satisfying the closure properties with respect to repeat functions $\mathbf{R}$ that is consistent with $\mathbf{T}\in\{f(n),\poly f(n)\}$.

    First note that for $\mathbf{T}(n) = f(n)$, an algorithm $\mathcal{A}$ with runtime $s(N)\in O(f(N))$ given $N$ qubit input has a runtime of $s(nt)\in O(f(n))$ given $t\in O(1)$ copies of $n$-qubit states.
    Then to show $f(n)$-indistinguishability we use a hybrid argument with: 
    \begin{enumerate}
        \item Hybrid 0: $t$ copies of size $|S|\in\omega(f(n))$ subset phase state ensemble $\{|\psi_{f,\sigma}\>\}_{f,\sigma}$ with $f(n)$-QPRPF $f$ and $f(n)$-QPRP $\sigma$ as an input to $\mathcal{A}$.
        \item Hybrid 1: $t$ copies of size $|S|\in\omega(f(n))$ subset phase state ensemble $\{|\psi_{\rand_\sigma,\rand_f}\>\}_{\rand_\sigma,\rand_f}$ for uniformly random permutation and phase function $\rand_\sigma,\rand_f$, respectively, as an input to $\mathcal{A}$.
        \item Hybrid 2: $t$ copies of Haar random ensemble $\{|\varphi\>\}$ as an input to $\mathcal{A}$.
    \end{enumerate}
    Clearly, for $t\in O(1)$ algorithm $\mathcal{A}$ outputs $\mathcal{A}(|\tau\>^{\otimes t})$ in $s(nt)\in O(f(n))$ since the input size is a just constant multiple of $n$.
    Now we show that
    \begin{equation}
        \Big|\Pr_{f,\sigma}[\algo(|\psi_{f,\sigma}\>^{\otimes t(n)})=1] - \Pr_\varphi[\algo(|\varphi\>^{\otimes t(n)})=1]\Big| < \eta(n) 
    \end{equation}
    for $\eta(n)\in\negl_{f(n)}$, namely that the Hybrid 0 and Hybrid 3 are $f(n)$-indistinguishable.
    We will use negligible function $\negl_{f(n)}$ with respect to repeat function $\mathbf{R}=O(1)$.
    Note that $\mathbf{R}=O(1)$ is consistent with $O(f(n))$ since for any $s(n)\in O(f(n))$ and any $r(n)\in O(1))$ it holds that $r(n)s(n) \in O(f(n))$.
    
    Now note that hybrid 0 and hybrid 1 are $O(f(n))$-indistinguishable since random permutation $\rand_\sigma$ is indistinguishable from $O(f(n))$-PRP $\sigma$ to all algorithms running in $O(f(n))$ and random function $\rand_f$ is indistinguishable from $O(f(n))$-PRP $\sigma$ to all algorithms running in $O(f(n))$, i.e.
    \begin{equation}\label{eqn:subset_phase_state_hybrid01_indistinguishability}
        \Big|\Pr_{f,\sigma}[\algo(|\psi_{f,\sigma}\>^{\otimes t})=1] - \Pr_{\rand_f,\rand_\sigma}[\algo(|\psi_{\rand_f,\rand_\sigma}\>^{\otimes t})=1]\Big| < \eta(n) 
    \end{equation}
    for some $\eta_0\in\negl_{f(n)}$.    
    Combining eqn.~\eqref{eqn:subset_phase_state_hybrid01_indistinguishability} above with part 1 of Proposition~\ref{prop:subset_phase_state_subset_size} that hybrid 1 and hybrid 2 are $f(n)$-indistinguishable:
    \begin{equation}\label{eqn:subset_phase_state_hybrid12_indistinguishability}
        \Big|\Pr_\varphi[\algo(|\varphi\>^{\otimes t})=1] - \Pr_{\rand_f,\rand_\sigma}[\algo(|\psi_{\rand_f,\rand_\sigma}\>^{\otimes t})=1]\Big| < \eta_1(n) 
    \end{equation}
    for some $\eta_1\in\negl_{f(n)}$, then by triangle inequality we have 
    \begin{equation}
    \begin{aligned}
        &\Big|\Pr_{f,\sigma}[\algo(|\psi_{f,\sigma}\>^{\otimes t(n)})=1] - \Pr_\varphi[\algo(|\varphi\>^{\otimes t(n)})=1]\Big| \\
        &\leq \Big|\Pr_{f,\sigma}[\algo(|\psi_{f,\sigma}\>^{\otimes t})=1] - \Pr_{\rand_f,\rand_\sigma}[\algo(|\psi_{\rand_f,\rand_\sigma}\>^{\otimes t})=1]\Big| \\
        &\quad + \Big|\Pr_\varphi[\algo(|\varphi\>^{\otimes t})=1] - \Pr_{\rand_f,\rand_\sigma}[\algo(|\psi_{\rand_f,\rand_\sigma}\>^{\otimes t})=1]\Big| \\
        &< \eta_0(n)+\eta_1(n)
    \end{aligned}
    \end{equation}
    since $\eta_0,\eta_1\in\negl_{f(n)}$, by the first closure property (Definition~\ref{def:closure_property}) of $\negl_{f(n)}$ it holds that $\eta_0(n)+\eta_1(n)\in\negl_{f(n)}$.

    For $\mathbf{T}(n)=\poly f(n)$, the proof is identical to the $\mathbf{T}(n)=f(n)$ case above.
    First, an algorithm $\mathcal{A}$ with runtime $s(N)\in O(\poly f(N))$ given $N$ qubit input has a runtime of $s(nt(n))\in O(\poly f(n))$ given $t(n)\in O(\poly f(n))$ copies of $n$-qubit states, since $s(nt(n)) = \poly(f(n\poly f(n)))$ is a polynomial since $f$ does not grow faster than polynomials.
    Then to show $\poly f(n)$-indistinguishability we use the same hybrid argument as above, but with $\poly f(n)$-QPRPF $f$ and $\poly f(n)$-QPRP $\sigma$, and number of copies $t(n)\in O(\poly f(n))$.
    Here we use negligible functions $\negl_{\poly f(n)}$ and repetition function $\mathbf{R}=O(\poly f(n))$ which is consistent with $\poly f(n)$ since for any $r(n),s(n)\in O(\poly f(n))$ it holds that $r(n)s(n)\in O(\poly f(n))$ again because $f$ does not grow faster than polynomials.
    
    As the case for $\mathbf{T}(n)=f(n)$ above, we can show that hybrid 0 with subset phase state input $|\psi_{f,\sigma}\>$ and hybrid 1 with $|\psi_{\rand_f,\rand_\sigma}\>$ input (both with $t(n)\in O(\poly f(n))$ copies thereof) are $\poly f(n)$-indistinguishable since we are using $\poly f(n)$-QPRPF $f$ and $\poly f(n)$-QPRP $\sigma$.
    Hybrid 1 and Hybrid 2 are also $\poly f(n)$-indistinguishable by Proposition~\ref{prop:subset_phase_state_subset_size}.
    Thus we can show that subset phase state ensemble $\{|\psi_{f,\sigma}\>\}_{f,\sigma}$ and Haar-random ensemble $\{|\varphi\>\}$ are $\poly f(n)$-indistinguishable by using triangle inequality and the closure property of $\negl_{\poly f(n)}$.
\end{proof}

\subsection{T-pseudorandom subset states}\label{sec:subset_states}

In this section we will give the subset state $\mathbf{T}$-PRS construction.

\begin{definition}
    An $n$-qubit \textit{subset state} $|S\>$ for $S\subseteq\{0,1\}^n$ is given by
    \begin{equation}
        |S\> = \frac{1}{\sqrt{|S|}} \sum_{x\in S} |x\> \;.
    \end{equation}
\end{definition}

Note that a subset state is similar to the subset phase state construction in Section~\ref{sec:subset_phase_states} in that we take the uniform superposition of $n$-bit strings in a subset $S\subseteq\{0,1\}^n$.
However, all of the individual terms here are phaseless.
This implies that the generator for an $n$-qubit subset state can be also constructed by the $O(\poly n)$ generator circuit of subset phase state in Lemma~\ref{lem:subset_phase_state_generator} using $\mathbf{T}$-PRP, but without the phase oracle $U^f$.
We denote this construction of subset state as $\{|S_\sigma\>\}_\sigma$ for $\mathbf{T}$-PRP $\sigma$.

Similar to the subset phase state, the lemma below gives an upper bound to the trace distance between an $n$-qubit random subset state ensemble $\{|S\>\}_S$ over all subsets of size $|S|=m$ and the $n$-qubit Haar-random ensemble $\{|\varphi\>\}_\varphi$.

\begin{lemma}[{\cite[Theorem 1]{giurgica2023pseudorandomness}}]\label{lem:trace_distance_random_subset_state}
    For subset $S\subseteq\{0,1\}^n$ with $|S|=m$ and for some positive integers $n,t$ it holds that
    \begin{equation}\label{eqn:trace_distance_random_subset_state}
        d_\mathrm{Tr}\bigg( \E_S[\ketbra{S}^{\otimes t}] \,,\, \E_\varphi[\ketbra{\varphi}^{\otimes t}] \bigg) \leq O\Big(\frac{tm}{2^n}\Big) + O\Big(\frac{t^2}{m}\Big)
    \end{equation}
    where subset $S$ is uniformly sampled from all possible $\binom{2^n}{m}$ size-$m$ subsets of $\{0,1\}^n$ and $\varphi$ is Haar-random.
\end{lemma}

Now we show how to determine the size of subset $S\subseteq\{0,1\}^n$ to construct a $\mathbf{T}$-PRS.

\begin{proposition}\label{prop:subset_state_subset_size}
    Let $f:\N\rightarrow\N$ be a non-decreasing function that grows at most polynomially.
    It holds that:
    \begin{enumerate}
        \item For number of copies $t:=t(n)\in O(1)$ and subset size $|S|=m(n)$ satisfying $\omega(f(n))< m(n)< o(2^n)$, the trace distance in eqn.~\eqref{eqn:trace_distance_random_subset_state} is $f(n)$-negligible.
        
        \item For number of copies $t:=t(n)\in O(\poly f(n))$ and subset size $|S|=m(n)$ satisfying $\omega(\poly f(n)) < m(n) < o(2^n)$, the trace distance in eqn.~\eqref{eqn:trace_distance_random_subset_state} is $\poly f(n)$-negligible.
    \end{enumerate}
\end{proposition}
\begin{proof}
    First, set $t:=t(n)=O(1)$ and $m:=m(n)$ such that $\omega(f(n)) < m(n) < o(2^n)$.
    We will evaluate the first term $O(tm/2^n)$ of the upper bound in Lemma~\ref{lem:trace_distance_random_subset_state}.
    Note that for all $c>0$ there exists $N$ such that $n\geq N\Rightarrow m(n)<c2^n$.
    Therefore for any $g(n)\in O(t(n)m(n)/2^n)$, it must hold that $g(n)\in O(2^{-n})$.
    
    Secondly, for the second term of the $O(t^2/m)$ upper bound in Lemma~\ref{lem:trace_distance_random_subset_state}, for all $c>0$ there exists $N\in\N$ such that $n\geq N \Rightarrow 1/m(n) < cf(n)$ since $m(n)\in\omega(f(n))$. 
    Thus for all $c>0$ and some constant $c'$ it holds that
    \begin{equation}
        \frac{t(n)^2}{m(n)} < \frac{t}{cf(n)}
    \end{equation}
    for all but finitely many $n$.
    Therefore for all $h(n)\in O(t^2/m)$, it holds that $h(n) \in o(f(n)^{-1})$.

    Putting both terms together, by Lemma~\ref{lem:trace_distance_random_subset_state} we obtain
    \begin{equation}
        O\Big(\frac{tm}{2^n}\Big) + O\Big(\frac{t^2}{m}\Big) < O\Big(\frac{1}{2^n}\Big) + o\Big(\frac{1}{f(n)}\Big) 
    \end{equation}
    which implies that 
    \begin{equation}
        d_\mathrm{Tr}\bigg( \E_S[\ketbra{S}^{\otimes t}] \,,\, \E_\varphi[\ketbra{\varphi}^{\otimes t}] \bigg) \in \negl_{O(f(n))} \;,
    \end{equation}
    hence proving the claim.

    Now for $t:=t(n)\in O(\poly f(n))$ and $\omega(\poly f(n)) < m(n) < o(2^n)$, we first look at the $O(tm/2^n)$ term.
    Note that for all $c>0$ and for some $c'$ it holds that $tm = t(n)m(n) < f(n)^{c'} 2^{cn}$ for all but finitely many $n$.
    Thus for all $g(n)\in O(tm/2^n)$ it holds that $g(n)\in O(2^{-n})$.
    Now for the $O(t^2/m)$ term, note that for some $c'>0$ and for all $c>0$ it holds that 
    \begin{equation}
        \frac{t(n)^2}{m(n)} < \frac{f(n)^{c'}}{f(n)^c}
    \end{equation}
    for all but finitely may $n$.
    Since $c>0$ can be arbitrarily large, therefore it holds that for any $g(n)\in O(t^2/m)$ we have $g(n)\in o(\frac{1}{\poly f(n)})$.
    Finally putting both terms together we have
    \begin{equation}
        O\Big(\frac{tm}{2^n}\Big) + O\Big(\frac{t^2}{m}\Big) < O(2^{-n}) + o\Big(\frac{1}{\poly f(n)}\Big) \;.
    \end{equation}
    Since the right hand side of the inequality is a $\poly f(n)$-negligible function $\negl_{\poly f(n)}$, thus the trace distance in eqn.~\eqref{eqn:trace_distance_random_subset_state} is $\poly f(n)$-negligible.
\end{proof}

Note that here we use similar $t(n)$ as in Proposition~\ref{prop:subset_phase_state_subset_size} for subset phase state for the same reasoning (see Remark~\ref{rem:number_of_copies}).

Finally, by using Proposition~\ref{prop:subset_state_subset_size} and the same $O(n)$ construction as the subset phase state (without the phase oracle) we obtain a subset state $\mathbf{T}$-PRS construction.

\begin{theorem}\label{thm:subset_state_Tpseudorandomness}
    Let $f:\N\rightarrow\N$ be a non-decreasing function that grows at most polynomially.
    It holds that:
    \begin{enumerate}
        \item A subset state ensemble $\{|\psi_\sigma\>\}_\sigma$ with 
        subset size $|S|=m(n)$ such that $\omega(f(n))< m(n)< o(2^n)$ is a $f(n)$-PRS given number of copies $t:=t(n)\in O(1)$.
        
        \item A subset state ensemble $\{|\psi_\sigma\>\}_\sigma$ with 
        subset size $|S|=m(n)$ such that $\omega(\poly f(n)) < m(n) < o(2^n)$ is a $\poly f(n)$-PRS given number of copies $t:=t(n)\in O(\poly f(n))$.
    \end{enumerate}
\end{theorem}
\begin{proof}
    Since $\{|\psi_\sigma\>\}_\sigma$ can be generated by an $O(n)$ circuit similar to the subset phase state construction, we only need to show its $\mathbf{T}$-indistinguishability from Haar-random ensemble $\{|\varphi\>\}$ for negligible functions $\negl_\mathbf{T}$ satisfying the closure properties with respect to repeat functions $\mathbf{R}$ that is consistent with $\mathbf{T}\in\{f(n),\poly f(n)\}$.

    The $\mathbf{T}$-indistinguishability proof for subset states is similar to that of subset phase state in Theorem~\ref{thm:subset_phase_state_Tpseudorandomness}.
    First for $\mathbf{T}(n) = f(n)$, an algorithm $\mathcal{A}$ with runtime $s(N)\in O(f(N))$ given $N$ qubit input has a runtime of $s(nt)\in O(f(n))$ given $t\in O(1)$ copies of $n$-qubit states.
    We use a hybrid argument with: 
    \begin{enumerate}
        \item Hybrid 0: $t$ copies of size $|S|=m(n)$ such that $\omega(f(n))< m(n)< o(2^n)$ subset state ensemble $\{|\psi_\sigma\>\}_\sigma$ with $f(n)$-QPRP $\sigma$ as an input to $\mathcal{A}$.
        \item Hybrid 1: $t$ copies of size $|S|=m(n)$ such that $\omega(f(n))< m(n)< o(2^n)$ subset phase state ensemble $\{|\psi_{\rand_\sigma}\>\}_{\rand_\sigma}$ for uniformly random permutation $\rand_\sigma$ as an input to $\mathcal{A}$.
        \item Hybrid 2: $t$ copies of Haar random ensemble $\{|\varphi\>\}$ as an input to $\mathcal{A}$.
    \end{enumerate}
    Clearly, for $t\in O(1)$ algorithm $\mathcal{A}$ outputs $\mathcal{A}(|\tau\>^{\otimes t})$ in $s(nt)\in O(f(n))$ since the input size is a just constant multiple of $n$.
    We use negligible function $\negl_{f(n)}$ with respect to repeat function $\mathbf{R}=O(1)$, which is consistent with $O(f(n))$ since for any $s(n)\in O(f(n))$ and any $r(n)\in O(1))$ it holds that $r(n)s(n) \in O(f(n))$.
    
    Now note that hybrid 0 and hybrid 1 are $O(f(n))$-indistinguishable since random permutation $\rand_\sigma$ is indistinguishable from $O(f(n))$-PRP $\sigma$ to all algorithms running in $O(f(n))$ and random function $\rand_f$ is $f(n)$-indistinguishable from $O(f(n))$-PRP $\sigma$ to all algorithms running in $O(f(n))$.   
    Along with part 1 of Proposition~\ref{prop:subset_phase_state_subset_size} that hybrid 1 and hybrid 2 are $f(n)$-indistinguishable, then by triangle inequality we have that
    \begin{equation}
    \begin{aligned}
        &\Big|\Pr_\sigma[\algo(|\psi_\sigma\>^{\otimes t(n)})=1] - \Pr_\varphi[\algo(|\varphi\>^{\otimes t(n)})=1]\Big| \\
        &\quad< \eta_0(n)+\eta_1(n)
    \end{aligned}
    \end{equation}
    for some $\eta_0,\eta_1\in\negl_{f(n)}$.
    By the first closure property (Definition~\ref{def:closure_property}) of $\negl_{f(n)}$ it holds that $\eta_0(n)+\eta_1(n)\in\negl_{f(n)}$.

    For $\mathbf{T}(n)=\poly f(n)$, the proof is identical to the $\mathbf{T}(n)=f(n)$ case above.
    First, an algorithm $\mathcal{A}$ with runtime $s(N)\in O(\poly f(N))$ given $N$ qubit input has a runtime of $s(nt(n))\in O(\poly f(n))$ given $t(n)\in O(\poly f(n))$ copies of $n$-qubit states, since $s(nt(n)) = \poly(f(n\poly f(n)))$ is a polynomial since $f$ does not grow faster than polynomials.
    Then to show $\poly f(n)$-indistinguishability we use the same hybrid argument, but with $\poly f(n)$-QPRP $\sigma$ and number of copies $t(n)\in O(\poly f(n))$.
    Here we use negligible functions $\negl_{\poly f(n)}$ and repetition function $\mathbf{R}=O(\poly f(n))$ which is consistent with $\poly f(n)$ since for any $r(n),s(n)\in O(\poly f(n))$ it holds that $r(n)s(n)\in O(\poly f(n))$ again because $f$ does not grow faster than polynomials.
    
    As the case for $\mathbf{T}(n)=f(n)$ above, we can show that hybrid 0 with subset state input $|\psi_\sigma\>$ and hybrid 1 with $|\psi_{\\rand_\sigma}\>$ input (both with $t(n)\in O(\poly f(n))$ copies thereof) are $\poly f(n)$-indistinguishable since we are using $\poly f(n)$-QPRPF $f$ and $\poly f(n)$-QPRP $\sigma$.
    Hybrid 1 and Hybrid 2 are also $\poly f(n)$-indistinguishable by Proposition~\ref{prop:subset_state_subset_size}.
    Thus we can show that subset state ensemble $\{|\psi_\sigma\>\}_{\sigma}$ and Haar-random ensemble $\{|\varphi\>\}$ are $\poly f(n)$-indistinguishable by using triangle inequality and the closure property of $\negl_{\poly f(n)}$.
\end{proof}

\section{T-Pseudoresources}\label{sec:T-pseudoresources}

While pseudorandomness alludes to how true randomness can be mimicked using lesser amount of randomness, pseudoresources indicates how objects possessing large amount of resources can be mimicked by those with small amount of resources~\cite{aaronson2022quantum,haug2023pseudorandom,bansal2024pseudorandom,Gu_2024}.
In the quantum regime, the study of pseudoresources show how quantum states with high amount of quantum resources such as coherence, entanglement, and magic can be substituted by states with low resource.
This is done mainly by using computational indistinguishability between two ensembles of states as we have discussed in Section~\ref{sec:pseudorandomness_and_indistinguishability}.
However, so far only polynomial-time indistinguishability has been studied with respect to pseudoresources.
As we have seen so far on how the polynomial-time bounded observers can be replaced by observers which computational runtime is bounded by some class of function $\mathbf{T}$, it is natural to do this generalization to pseudoresources as well.

For a given resource, we can assign a resource measure\footnote{Specifically, a resource measure $Q$ usually is required to satisfy certain properties. The most common required properties are: (1) Faithfulness: $Q$ must assign a 0 value to a prescribed set of ``free states'' $\mathcal{F}$, i.e. $Q(\rho)=0, \forall\rho\in\mathcal{F}$, and (2) Monotonicity: $Q$ must satisfy $Q(\rho) \geq Q(\mathcal{C}(\rho))$ for any state $\rho$ and any $\mathcal{C}$ belonging to a prescribed set of ``free operations'' $\mathcal{O}$. Other nice properties of $Q$ such as convexity, subadditivity, and continuity could also be demanded. This is part of the study of quantum \textit{resource theories} which we will not go into detail. We will instead use resource measures that are commonly used in the literature. Readers who are interested to find out more about resource theory may refer to~\cite{chitambar2019quantum,gour2024resources}.} $Q$ which assigns a (real-number) value $Q(\psi)$ to a quantum state $|\psi\>$.
Note that here we only consider pure quantum states.
Furthermore, for a resource measure $Q$ and quantum state ensembles $\{|\varphi\>\}$ and $\{|\psi\>\}$, we define the \textit{resource gap} between $\{|\varphi\>\}$ and $\{|\psi\>\}$ as
\begin{equation}
    \Delta_Q(\{|\varphi\>\},\{|\psi\>\}) := \Big| \E_\varphi[Q(\varphi)] - \E_\psi[Q(\psi)] \Big| \;.
\end{equation}
For computationally indistinguishable ensembles $\{|\varphi\>\}$ (as in Definition~\ref{def:T_indistinguishable} and $\{|\psi\>\}$ where expected resource $\E_\varphi[Q(\varphi)]$ is larger than the expected resource $\E_\psi[Q(\psi)]$, this indicates that ensemble $\{|\psi\>\}$ acts as a \textit{pseudoresource}, mimicking the high-resource ensemble $\{|\varphi\>\}$ with respect to some computationally-bounded observer.

As we will see later in this section, we can use $\mathbf{T}$-PRS to obtain larger resource gaps for coherence (Table~\ref{tab:observer_computational_power_coherence_gap}), entanglement (Table~\ref{tab:observer_computational_power_entanglement_gap}), and magic (Table~\ref{tab:observer_computational_power_magic_gap}), compared to the usual pseudoresource gap from polynomial-time PRS.
While results on pseudorandom density matrices (PRDM) in~\cite{bansal2024pseudorandom} has shown that the largest amount of resource gap can be obtained from a mixed-state generalization of polynomial-time PRS, i.e. $\E_\psi[Q(\psi)]=0$, $\mathbf{T}$-PRS give intermediate resource gaps between those obtained from polynomial-time PRS and PRDM.

\subsection{Coherence resource gap}\label{sec:coherence}

For an $n$-qubit state $\rho$, the relative entropy of coherence~\cite{baumgratz2014quantifying,streltsov2017colloquium} of $\rho$ is defined as
\begin{equation}
    C(\rho) = H(\rho_\mathrm{diag}) - H(\rho) \;,
\end{equation}
which takes value between $0$ and $n$.
Here $H(\rho) = -\tr(\rho \log\rho)$ is the von Neumann entropy of density matrix $\rho$.
Relative entropy of coherence of $\rho$ admits an operational interpretation as the asymptotic rate of how many copies single-qubit maximally coherent state can be obtained for every $\rho$ to distill a  using incoherent operations (see~\cite{baumgratz2014quantifying},\cite[Section III.C]{streltsov2017colloquium}).

Now let us consider the Hilbert-Schmidt coherence distance of quantum state $\rho$, given by
\begin{equation}
    C_2(\rho) = \min_{\sigma\in\mathcal{F}_c} ||\rho-\sigma||_{\mathrm{HS}}^2 = 1 - \tr(\rho^2 \Pi_{c2})
\end{equation}
for projector $\Pi_{c2} = \bigotimes_{j=1}^n \ketbra{00}+\ketbra{11}$.
For all state $\rho$, the Hilbert-Schmidt coherence distance $C_2(\rho)$ takes value between $0$ and $1-2^{-n}$ and satisfy the relation
\begin{equation}\label{eqn:HS_coherence_RE_coherence_relation}
    C(\rho) \geq -\log(1-C_2(\rho)) \;.
\end{equation}

\begin{proposition}\label{prop:rel_entropy_coherence_gap}
    For any $\mathbf{T}$-indistinguishable ensembles $\{|\varphi\>\}$ and $\{|\psi\>\}$ such that $\E_\varphi[C_2(|\varphi\>)] \geq 1-2^{-\gamma(n)} \geq \E_\psi[C_2(|\psi\>)]$ for some function $\gamma:\N\rightarrow\N$, it holds that
    \begin{equation}\label{eqn:coherence_lowerbound}
    \begin{aligned}
        \E_\psi[C(|\psi\>)] &\geq -\log\bigg(\frac{1}{2^{\gamma(n)}}+\negl_\mathbf{T}(n)\bigg) \;.
    \end{aligned}
    \end{equation}
    and
    \begin{equation}\label{eqn:coherence_gap}
        \Delta_C(\{|\varphi\>\},\{|\psi\>\}) = O(n) +\log\bigg(\frac{1}{2^{\gamma(n)}}+\negl_\mathbf{T}(n)\bigg) \;.
    \end{equation}
\end{proposition}
\begin{proof}
    We can use the projector $\Pi_{c2}$ as an efficient distinguisher with acceptance probability of $p(\rho)=\tr(\rho^{\otimes2}\Pi_{c2})$ given two copies of $\rho$ as input.
    The expected average acceptance probability for $\{|\varphi\>\}$ is $\E_\varphi[p(\varphi)] \leq 2^{-\gamma(n)}$, whereas the average acceptance probability for $\{|\psi\>\}$ is $\E_\psi[p(\psi)]=1-\E_\psi[C_2(|\psi\>)] \geq 2^{-\gamma(n)}$.
    Since these ensembles are computationally indistinguishable and $\E_\psi[p(\psi)] \geq 2^{-\gamma(n)} \geq \E_\varphi[p(\varphi)]$, it holds that
    \begin{equation}
        \E_\psi[p(\psi)] - \E_{\varphi}[p(\varphi)] = \eta(n) \;,
    \end{equation}
    for $\eta\in\negl_\mathbf{T}$, which implies that
    \begin{equation}
    \begin{aligned}
        \E_\psi[p(\psi)] &= \E_\varphi[p(\varphi)] + \eta(n) \\ 
        &\geq 2^{-\gamma(n)} + \eta(n) \;.
    \end{aligned}
    \end{equation}
    Now by using the relation $C(\rho) \geq -\log(1-C_2(\rho))$ (eqn.~\eqref{eqn:HS_coherence_RE_coherence_relation}) we obtain
    \begin{equation}
    \begin{aligned}
        \E_\psi[C(|\psi\>)] &\geq \E_\psi[-\log(1-C_2(|\psi\>))] \\
        &= \E_\psi[-\log p(\psi)] \\
        &= -\log\bigg(\frac{1}{2^{\gamma(n)}}+\eta(n)\bigg) \;.
    \end{aligned}
    \end{equation}
    which gives us eqn.~\eqref{eqn:coherence_lowerbound}.
    Lastly by observing that the maximum value of the relative entropy of coherence is $\max_\rho C(\rho) = O(n)$ and combining it with eqn.~\eqref{eqn:coherence_lowerbound} we get
    \begin{equation}
    \begin{aligned}
        &\Delta_C(\{|\varphi\>\},\{|\psi\>\}) \\
        &= \E_\varphi[C(\varphi)] + \log\bigg(\frac{1}{2^{\gamma(n)}}+\eta(n)\bigg) \\
        &= O(n) + \log\bigg(\frac{1}{2^{\gamma(n)}}+\eta(n)\bigg)
    \end{aligned}
    \end{equation}
    which gives us eqn.~\eqref{eqn:coherence_gap}.
\end{proof}

\begin{table*}[]
    \centering
    \[
    \begin{array}{|c|c|c|}
        \hline
        \mathbf{T} & \E_\psi[C(\psi)] & \Delta_C(\{\varphi\},\{\psi\}) \leq \\
        \hline
        O(\poly n) & \omega(\log n) & O(n) - \omega(\log n) \\
        O(n\log n) & \omega(1) + \log(n\log n) & O(n) - (\omega(1) + \log(n\log n)) \\
        O(n) & \omega(1) + \log n & O(n) - (\omega(1) + \log n) \\
        O(\poly\log n) & \omega(\log\log n) & O(n) - \omega(\log\log n) \\
        O(\log n) & \omega(1) + \log\log n & O(n) - (\omega(1) + \log\log n) \\
        \hline
    \end{array}
    \]
    \caption{
    Left column: Computational power $\mathbf{T}$ of the observer. 
    Center column: Expected relative entropy of coherence of ensemble $\{|\psi\>\}$ that is indistinguishable from $\{|\varphi\>\}$ to $\mathbf{T}$-time observers. 
    Right column: Upper bound of coherence gap $\Delta_C$ between Haar-random ensemble $\{|\varphi\>\}$ and ensemble $\{|\psi\>\}$.
    The quantities in the center and right columns can be obtained directly from Proposition~\ref{prop:rel_entropy_coherence_gap} by setting $2^{-\gamma(n)} = O(2^{-n})$ and setting $\mathbf{T}(n)$ as in the left column.
    Note that as the computational power of the observer increases, one can see that the average amount of coherence of the pseudo-random ensemble $\E_\psi[C(\psi)]$ decreases, so that the pseudo-coherence gap increases.
    Namely, it cost less coherence to fool a computationally weaker observer.
    }
    \label{tab:observer_computational_power_coherence_gap}
\end{table*}

If we set ensemble $\{|\varphi\>\}$ as the Haar-random state, its expected coherence and expected Hilbert-Schmidt coherence distance are
\begin{equation}
\begin{gathered}
    \E_\varphi[C(\varphi)] = \sum_{k=2}^{2^n}\frac{1}{k} = O(n) \\
    \text{and}\\
    \E_\varphi[C_2(\varphi)] = 1-\frac{2}{2^n+1} \geq 1-O(2^{-n}) \;,
\end{gathered}
\end{equation}
respectively.
The following are the expected coherence of ensemble $\{|\psi\>\}$ that is $\mathbf{T}$-indistinguishable to the Haar-random ensemble $\{|\varphi\>\}$ to observers with different computational power $\mathbf{T}$.
The results are summarized in Table~\ref{tab:observer_computational_power_coherence_gap}.
\begin{enumerate}
    \item For a $\poly$-time observer (i.e. $\mathbf{T}(n) = \poly n = n^{O(1)}$), it holds that
    \begin{equation}
    \begin{aligned}
        O(2^{-n})+\frac{1}{\poly n} = 2^{-O(\log n)} \;,
    \end{aligned}
    \end{equation}
    since $\poly n = n^{O(1)} = 2^{O(1)\log n} = 2^{O(\log n)}$ and since $O(2^{-n})$ grows slower than $2^{-\omega(\log n)}$.
    Thus, by Proposition~\ref{prop:rel_entropy_coherence_gap}, we obtain
    \begin{equation}
        \E_\psi[C(\psi)] \geq -\log( 2^{-O(\log n)} ) = \omega(\log n) \;,
    \end{equation}
    which agrees with the bound in~\cite[Appendix I]{haug2023pseudorandom}.
    
    \item For a linearithmic time observer (i.e. $\mathbf{T}(n) = O(n\log n)$, first note that for $g(n) \in O(n\log n)$, it holds that there exists $c>0$ and $N\in\N$ such that $g(n) < cn\log n$ if $n>N$, which is equivalent to $g(n) \in O(1) n\log n$.
    Thus we can obtain the equivalence $O(n\log n) = 2^{O(1) + \log(n\log n)}$, which gives
    \begin{equation}
    \begin{aligned}
        &O(2^{-n}) + \frac{1}{O(n\log n)} \\
        &= O(2^{-n}) + 2^{-O(1) - \log(n\log n)} \\
        &= 2^{-O(1) - \log(n\log n)} \;,
    \end{aligned}
    \end{equation}
    as any $g(n) \in O(2^{-n})$ also satisfy $g(n)\in 2^{-O(1) - \log(n\log n)}$.
    Thus by Proposition~\ref{prop:rel_entropy_coherence_gap}, the expected relative entropy of coherence of $\{|\psi\>\}$ is lower bounded as
    \begin{equation}
    \begin{aligned}
        \E_\psi[C(\psi)] &\geq -\log\big(2^{-O(1) - \log(n\log n)}\big) \\
        &= \omega(1) + \log(n\log n) \;.
    \end{aligned}
    \end{equation}
    
    \item For a linear-time observer (i.e. $\mathbf{T}(n)=O(n)$), by using the equivalence $O(n) = 2^{O(1) + \log n}$ we have
    \begin{equation}
    \begin{aligned}
        O(2^{-n}) + \frac{1}{O(n)} = 2^{-O(1)-\log n} \;.
    \end{aligned}
    \end{equation}
    Thus the expected relative entropy of coherence of $\{|\psi\>\}$ is lower bounded as
    \begin{equation}
    \begin{aligned}
        \E_\psi[C(\psi)] &\geq -\log\Big(2^{-O(1)-\log n}\Big) \\
        &= \omega(1) + \log n \;.
    \end{aligned}
    \end{equation}

    \item For polylogarithmic time observer (i.e. $\mathbf{T}(n) = O(\poly\log(n))$), first note that $\poly\log n = \log^{O(1)}n = 2^{\log(\log^{O(1)} n)} = 2^{O(\log\log n)}$.
    Hence we obtain
    \begin{equation}
    \begin{aligned}
        2^{-\gamma(n)} + \mathbf{T}(n)^{-1} &= O(2^{-n}) + 2^{-O(\log\log n)} \\
        &= 2^{-O(\log\log n)} \;.
    \end{aligned}
    \end{equation}
    Then by using Proposition~\ref{prop:rel_entropy_coherence_gap} this gives
    \begin{equation}
    \begin{aligned}
        \E_\psi[C(|\psi\>)] &\geq -\log\Big( 2^{-O(\log\log n)} \Big) \\
        &= \omega(\log\log n) \;.
    \end{aligned}
    \end{equation}
    
    \item For logarithmic time observer (i.e. $\mathbf{T}(n) = O(\log n)$), we have the equivalence $O(\log n) = 2^{O(1) + \log\log n}$ which gives
    \begin{equation}
    \begin{aligned}
        O(2^{-n}) + 2^{-O(1) - \log\log n} = 2^{-O(1) - \log\log n} \;,
    \end{aligned}
    \end{equation}
    since $O(2^{-n})$ grows slower than $2^{-O(1) - \log\log n}$.
    Thus the expected relative entropy of coherence of $\{|\psi\>\}$ is lower bounded as
    \begin{equation}
    \begin{aligned}
        \E_\psi[C(\psi)] &\geq -\log\big(2^{-O(1) - \log\log n}\big) \\
        &= \omega(1) + \log\log n \;.
    \end{aligned}
    \end{equation}
\end{enumerate}

\subsection{Entanglement resource gap}\label{sec:entanglement}

\begin{table*}[]
    \centering
    \[
    \begin{array}{|c|c|c|}
        \hline
        \mathbf{T}(n) & \E_\psi[E(\psi)] & \Delta_E(\{\varphi\},\{\psi\}) \leq \\
        \hline
        O(\poly n) & \omega(\log n) & O(n) - \omega(\log n) \\
        O(n\log n) & \omega(1) + \log(n\log n) & O(n) - (\omega(1) + \log(n\log n)) \\
        O(n) & \omega(1) + \log n & O(n) - (\omega(1) + \log n) \\
        O(\poly\log n) & \omega(\log\log n) & O(n) - \omega(\log\log n) \\
        O(\log n) & \omega(1) + \log\log n & O(n) - (\omega(1) + \log\log n) \\
        \hline
    \end{array}
    \]
    \caption{
    Left column: Computational power $\mathbf{T}$ of the observer. 
    Center column: Expected entanglement entropy $\E_\psi[E(\psi)]=f(n)$ of $n$-qubit ensemble $\{|\psi\>\}$ over partition $A:B$ with $\dim A=2^{n_A}$ and $\dim B=2^{n_B}$ with $n_A\leq n_B$ and $n_A\in\Omega(f(n))$ and $n_A+n_B=n$ of ensemble $\{|\psi\>\}$ that is indistinguishable from $\{|\varphi\>\}$ to a $\mathbf{T}$-observer.
    Right column: Upper bound of entanglement gap $\Delta_E$ between Haar-random ensemble $\{|\varphi\>\}$ and ensemble $\{|\psi\>\}$.
    The quantities in the center and right columns can be obtained directly from Proposition~\ref{prop:entanglement_gap} by setting $2^{-\gamma(n)} = O(2^{-n})$ and setting $\mathbf{T}(n)$ as in the left column, with negligible functions $\negl_\mathbf{T}$.
    }
    \label{tab:observer_computational_power_entanglement_gap}
\end{table*}

Here we use entanglement entropy as a measure of how much entanglement does a quantum state has.
\textit{Entanglement entropy} of a bipartite $n$-qubit pure quantum state $|\psi\>$ over system partition $A\otimes B$ with dimensions $\dim A = 2^{n_A}$ and $\dim B=2^{n_B}$ and $n_A+n_B=n$ is given by
\begin{equation}
    E(\psi) = H(\tr_A(\ketbra{\psi})) = H(\tr_B(\ketbra{\psi}))
\end{equation}
where $\tr_A$ ($\tr_B$) denotes a partial trace on system $A$ (system $B$).
Note that the equality between the entropy of $|\psi\>$ reduces to system $A$ and $B$ follows from the fact that the entropy of bipartite $A:B$ pure quantum states reduced to either of the systems $A$ or $B$ is equal.

Operationally, it has been shown that the entanglement entropy of a bipartite pure state $|\psi\>$ correspond to both~\cite{bennett1996concentrating,berta2024tangled}: 
(1) the entanglement cost of $|\psi\>$, i.e. the asymptotic rate of how many two-qubit maximally entangled states $|\phi\>$ is needed to obatain a copy of $|\psi\>$ using only local operations and classical communications (LOCC) operations, and 
(2) the distillable entanglement of $|\psi\>$, i.e. the asymptotic rate of how many copies of $|\phi\>$ can be obtained per copy of $|\psi\>$ using LOCC operations.

A useful tool used to show a lower bound of entanglement entropy of PRS~\cite{ji2018pseudorandom,aaronson2022quantum} is the SWAP test.
The SWAP test takes two copies of quantum states $\rho$ and outputs ``accept'' with probability
\begin{equation}\label{eqn:swap_test_proba}
    p_\mathrm{SW}(\rho) = \frac{1}{2}(1+\tr(\rho^2)) = \frac{1}{2}(1+2^{-H_2(\rho)})
\end{equation}
where $H_2(\rho)=-\tr(\rho^2)$ is the quantum collision entropy.
Now we state our result characterizing the entanglement entropies of two indistinguishable ensemble of quantum states.

\begin{proposition}\label{prop:entanglement_gap}
    Consider a $\mathbf{T}$-computationally indistinguishable ensembles $\{|\varphi\>\}$ and $\{|\psi\>\}$ over $n$ qubit system partitioned as $A\otimes B$ where $\dim A=2^{n_A}$ and $\dim B = 2^{n_B}$ and $n_A\leq n_b$ and $n_A+n_B=n$, such that $\E_\varphi[H_2(\varphi_A)] \geq \xi(n) \geq \E_\psi[H_2(\psi_A)]$ for some function $\xi:\N\rightarrow\N$. 
    It holds that
    \begin{equation}
        \E_\psi[E(\psi)] \geq -\log\Big( \frac{1}{2^{\xi(n)}} + \eta(n) \Big)
    \end{equation}
    and
    \begin{equation}
        \Delta_E(\{|\varphi\>\}\,,\,\{|\psi\>\}) = O(n) + \log\Big( \frac{1}{2^{\xi(n)}} + \eta(n) \Big) \;,
    \end{equation}
    for some $\eta\in\negl_\mathbf{T}$.
    Here, by taking the expected entanglement entropy of $n$-qubit ensemble $\{|\psi\>\}$ as a function in $n$: $\E_\psi[E(\psi)]=f(n)$, we also assume that $n_A\in \Omega(f(n))$.
\end{proposition}

\begin{proof}
    First we use the fact that $H(\rho)\geq H_2(\rho)$ for all states $\rho$ and then express the collision entropy $H_2$ in terms of the accept probability of the SWAP test in eqn.~\eqref{eqn:swap_test_proba} to obtain
    \begin{equation}\label{eqn:vonneumann_entropy_renyi_swap_test_inequality}
    \begin{gathered}
        H(\rho) \geq H_2(\rho) = -\log( 2p_\mathrm{SW}(\rho) - 1 ) \;.
    \end{gathered}
    \end{equation}
    Indistinguishability between ensembles $\{|\psi\>\}$ and $\{|\varphi\>\}$ to a $\mathbf{T}$-bounded observer implies that
    \begin{equation}
        \big| \E_\psi[p_\mathrm{SW}(\psi_A)] - \E_\varphi[p_\mathrm{SW}(\varphi_A)] \big| = \eta(n) \;,
    \end{equation}
    for some $\eta\in\negl_\mathbf{T}$ by taking the SWAP test as a constant-size quantum circuit acting as a distinguisher.
    Since $\E_\varphi[H_2(\varphi_A)] \geq \xi(n) \geq \E_\psi[H_2(\psi_A)]$ by assumption, we have 
    \begin{equation}\label{eqn:swap_accept_upper_bound}
    \begin{aligned}
        \E_\psi[p_\mathrm{SW}(\psi_A)] &= \E_\varphi[p_\mathrm{SW}(\varphi_A)] + \eta(n) \\
        &= \frac{1}{2} + \E_\varphi[2^{-H_2(\varphi_A)-1}] + \eta(n) \\
        &\leq \frac{1}{2} + 2^{-\xi(n)-1} + \eta(n) \;.
    \end{aligned}
    \end{equation}
    
    By eqn.~\eqref{eqn:vonneumann_entropy_renyi_swap_test_inequality} and eqn.~\eqref{eqn:swap_accept_upper_bound} the average entanglement entropy of ensemble $\{|\psi\>\}$ can be lower bounded as
    \begin{equation}
    \begin{aligned}
        \E_\psi[E(\psi)] &= \E_\psi[H(\psi_A)] \\
        &\geq \E_\psi[H_2(\psi_A)] \\
        &= \E_\psi\big[-\log( 2p_\mathrm{SW}(\psi_A) - 1 )\big] \\
        &\geq -\log\Big( \frac{1}{2^{\xi(n)}} + \eta(n) \Big) \;,
    \end{aligned}
    \end{equation}
    which can be rewritten as
    \begin{equation}
        \E_\psi[E(\psi)] \geq -\log\Big( \frac{1}{2^{\xi(n)}} + \frac{1}{\mathbf{T}(n)} \Big)
    \end{equation}
    since $\eta(n)<\frac{1}{g(n)}$ for all $g\in\mathbf{T}$.
\end{proof}

If $\{|\varphi\>\}$ is a Haar-random ensemble, then its expected entanglement entropy and expected R\'enyi-2 entanglement entropy over partition $A:B$ with $\dim A=2^{n_A}$ and $\dim B=2^{n_B}$ with $n_A\leq n_B$ and $n_A+n_B=n$ is given by
\begin{equation}
\begin{gathered}
    \E_\varphi[E(\varphi)] = \min\{n_A,n_B\} - O(1) = n_A-O(1) \\
    \text{and}\\
    \E_\varphi[H_2(\varphi_A)] = -\log\Big(\frac{2^{n_A} + 2^{n_B}}{2^n+1}\Big) \;,
\end{gathered}
\end{equation}
which gives $\xi(n) \in O(n_A) \subseteq O(n)$.

Now we give a lower bound for expected entanglement entropy of $n$-qubit ensemble $\{|\psi\>\}$ with low-entanglement that is $\mathbf{T}$-indistinguishable from $n$-qubit Haar-random ensemble $\{|\varphi\>\}$ for different $\mathbf{T}$.
The results are summarized in Table~\ref{tab:observer_computational_power_entanglement_gap}.
The derivations are similar to that of relative entropy of coherence.
\begin{enumerate}
    \item For a poly-time observer (i.e. $O(\mathbf{T}) = O(\poly(n))$), it holds that for $\eta\in\negl_{\poly n}$
    \begin{equation}
        \frac{1}{2^{\xi(n)}} + \eta(n) < O(2^{-n})+\frac{1}{\poly n} = 2^{-O(\log n)} \;.
    \end{equation}
    Hence by Proposition~\ref{prop:entanglement_gap},
    \begin{equation}
    \begin{aligned}
        \E_\psi[E(\psi)] 
        &\geq -\log(2^{-O(\log n}) \\
        &= \omega(\log n) \;,
    \end{aligned}
    \end{equation}
    which matches the bound in~\cite{aaronson2022quantum}.

    \item For linearithmic-time ($O(n\log n)$) observer, it holds that for $\eta\in\negl_{n\log n}$
    \begin{equation}
    \begin{aligned}
        \frac{1}{2^{\xi(n)}} + \eta(n) &< O(2^{-n}) + \frac{1}{O(n\log n)} \\
        &= O(2^{-n}) + 2^{-O(1)-\log(n\log n)}
    \end{aligned}
    \end{equation}
    Hence by Proposition~\ref{prop:entanglement_gap},
    \begin{equation}
    \begin{aligned}
        \E_\psi[E(\psi)] 
        &\geq -\log\Big( O(2^{-n}) + 2^{-O(1)-\log(n\log n)} \Big) \\
        &= \omega(1) + \log(n\log n) \;.
    \end{aligned}
    \end{equation}

    \item For linear-time observer ($O(n)$), it holds that for $\eta\in\negl_{n}$
    \begin{equation}
    \begin{aligned}
        \frac{1}{2^{\xi(n)}} + \eta(n) < O(2^{-n}) + \frac{1}{O(n)} = 2^{-O(1)-\log n}
    \end{aligned}
    \end{equation}
    Hence by Proposition~\ref{prop:entanglement_gap},
    \begin{equation}
    \begin{aligned}
        \E_\psi[E(\psi)] 
        &\geq -\log\Big( O(2^{-n}) + 2^{-O(1)-\log n} \Big) \\
        &= \omega(1) + \log n \;.
    \end{aligned}
    \end{equation}

    \item For polylogarithmic-time observer ($O(\poly\log n)$), it holds that for $\eta\in\negl_{\poly\log n}$
    \begin{equation}
    \begin{aligned}
        \frac{1}{2^{\xi(n)}} + \eta(n) < O(2^{-n}) + 2^{-O(\log\log n)}
    \end{aligned}
    \end{equation}
    Hence by Proposition~\ref{prop:entanglement_gap},
    \begin{equation}
    \begin{aligned}
        \E_\psi[E(\psi)] 
        &\geq -\log\Big( O(2^{-n}) + 2^{-O(\log\log n)} \Big) \\
        &= \omega(\log\log n) \;.
    \end{aligned}
    \end{equation}

    \item For logarithmic-time observer ($O(\log n)$), it holds that for $\eta\in\negl_{\log n}$
    \begin{equation}
    \begin{aligned}
        \frac{1}{2^{\xi(n)}} + \eta(n) < O(2^{-n}) + 2^{-O(1)-\log\log n}
    \end{aligned}
    \end{equation}
    Hence by Proposition~\ref{prop:entanglement_gap},
    \begin{equation}
    \begin{aligned}
        \E_\psi[E(\psi)] 
        &\geq -\log\Big( O(2^{-n}) + 2^{-O(1)-\log\log n} \Big) \\
        &= \omega(1) + \log\log n \;.
    \end{aligned}
    \end{equation}
\end{enumerate}

\subsection{Magic resource gap}\label{sec:magic}

Stabilizer R\'enyi-$\alpha$ entropy~\cite{leone2022stabilizer,haug2023stabilizer} of $n$-qubit state $\rho$ is given by
\begin{equation}\label{eqn:stab_renyi_entropy}
    M_\alpha(\rho) = \frac{1}{1-\alpha} \log\bigg( \frac{1}{2^n}\sum_{P\in\mathcal{P}_n} \big(\tr(P\rho)\big)^{2\alpha} \bigg) \;,
\end{equation}
where $\mathcal{P}_n$ is the set of all $n$-qubit Paulis modulo phases $-I,\pm iI$.

Here we use the Hadamard test~\cite{haug2024efficient,Gu_2024} which uses $2\alpha$ copies (for odd $\alpha$) of $n$-qubit state $\rho$ accepts with probability
\begin{equation}
\begin{aligned}
    p_\mathrm{H}^{(2\alpha)}(\rho) = \frac{1+\tr(\Pi^{(2\alpha)}\rho^{\otimes 2\alpha})}{2}
\end{aligned}
\end{equation}
where $\Pi^{(2\alpha)} = \frac{1}{2^n}\sum_{P\in\mathcal{P}_n} P^{\otimes 2\alpha}$.
Note that we can express the stabilizer R\'enyi-$\alpha$ entropy in terms of $\Pi^{(2\alpha)}$ as
\begin{equation}
    M_\alpha(\rho) = \frac{1}{1-\alpha}\log \tr\Big(\Pi^{(2\alpha)} \rho^{\otimes 2\alpha}\Big) \;.
\end{equation}
Hence accepting probability of the Hadamard test using $2\alpha$ copies and the stabilizer R\'enyi-$\alpha$ entropy can be expressed in terms of one another as
\begin{equation}\label{eqn:stab_entropy_hadamard_test_correspondence}
\begin{gathered}
    M_\alpha(\rho) = \frac{1}{1-\alpha}\log\Big( 2p_\mathrm{H}^{(2\alpha)}(\rho) - 1 \Big) \\
    \text{and}\\
    p_\mathrm{H}^{(2\alpha)}(\rho) = \frac{1}{2}\Big( 1 + 2^{(1-\alpha)M_\alpha(\rho)} \Big) \;.
\end{gathered}
\end{equation}

\begin{proposition}\label{prop:stabilizer_entropy_gap}
    Let $\alpha\geq2$ be an odd integer with $\alpha=h(n)$ such that $s(nh(n))\in O(\mathbf{T}(n))$ for all $s\in O(\mathbf{T}(n))$.
    Then, for $\mathbf{T}$-computationally indistinguishable ensembles $\{|\varphi\>\}$ and $\{|\psi\>\}$ such that $\E_\varphi[M_\alpha(\varphi)] \geq \tau(n) \geq \E_\psi[M_\alpha(\psi)]$ for some function $\tau:\N\rightarrow\N$, it holds that
    \begin{equation}\label{eqn:stab_entropy_lowerbound}
    \begin{aligned}
        \E_\psi[M_\alpha(\psi)] &\geq -\frac{\log(\negl_\mathbf{T}(n)) + \frac{2^{-(\alpha-1)\tau(n)}}{\negl_\mathbf{T}(n)} }{\alpha-1} \;.
    \end{aligned}
    \end{equation}
    and
    \begin{equation}\label{eqn:stab_entropy_gap}
        \Delta_{M_\alpha}(\{|\varphi\>\},\{|\psi\>\}) \leq O(n) + \frac{\log(\negl_\mathbf{T}(n)) + \frac{2^{-(\alpha-1)\tau(n)}}{\negl_\mathbf{T}(n)} }{\alpha-1} \;.
    \end{equation}
\end{proposition}

\begin{proof}
    For $n$-qubit quantum state ensembles $\{|\psi\>\}$ and $\{|\varphi\>\}$ that are $\mathbf{T}$-indistinguishable, therefore for any quantum algorithm $\mathcal{A}$ with runtime bounded by $s\in O(\mathbf{T})$ it must hold that
    \begin{equation}
        \Big| \E_\psi[\mathcal{A}(\psi^{\otimes t(n)})=1] - \E_\varphi[\mathcal{A}(\varphi^{\otimes t(n)})=1] \Big| = \eta(n)
    \end{equation}
    for any $t(n)$ such that $s(nt(n))\in O(\mathbf{T}(n))$ and $\eta\in\negl_\mathbf{T}$.
    Thus if $\mathcal{C}$ is the Hadamard test circuit and $\alpha = t(n)/2$ we have
    \begin{equation}
        \Big| \E_\psi[p_\mathrm{H}^{(2\alpha)}(\psi)] - \E_\varphi[p_\mathrm{H}^{(2\alpha)}(\varphi)] \Big| = \eta(n) \;,
    \end{equation}
    for some $\eta\in\negl_\mathbf{T}$.
    Thus by eqn.~\eqref{eqn:stab_entropy_hadamard_test_correspondence} it holds that
    \begin{equation}
    \begin{aligned}
        \frac{1}{2}\Big| 2^{(1-\alpha)\E_\psi[M_\alpha(\psi)]} - 2^{(1-\alpha)\E_\varphi[M_\alpha(\varphi)]} \Big| &= \eta(n) \;.
    \end{aligned}
    \end{equation}
    Since $\E_\varphi[M_\alpha(\varphi)] \geq \tau(n) \geq \E_\psi[M_\alpha(\psi)]$ and $\alpha>1$ we have
    \begin{equation}\label{eqn:stabilizer_entropy_exponent_bound}
    \begin{aligned}
        2^{-(\alpha-1)\E_\psi[M_\alpha(\psi)]} &= 2^{-(\alpha-1)\E_\varphi[M_\alpha(\varphi)]} + 2\eta(n) \\
        &\leq 2^{-(\alpha-1)\tau(n)} + 2\eta(n) \;.
    \end{aligned}
    \end{equation}
    Note that since $\E_\psi[\tr(\Pi^{(2\alpha)}\psi^{2\alpha})] = 2^{(1-\alpha)\E_\psi[M_\alpha(\psi)]}$ this also puts a bound on $\E_\psi[\tr(\Pi^{(2\alpha)}\psi^{\otimes 2\alpha})]$ and $\E_\psi[p_H^{(2\alpha)}(\psi)]$.\footnote{
    In the proof of Lemma S1 of~\cite{Gu_2024} it is shown that $\tr(\Pi^{(2\alpha)}\psi^{\otimes 2\alpha}) \in o((\poly n)^{-1})$ whenever $\E[M_\alpha(\varphi)] \in \Omega(n)$ for $\eta(n)\in\negl_\poly(n)$ (which is true for Haar-random ensemble $\{|\varphi\>\}$). 
    This can be obtained from eqn.~\eqref{eqn:stabilizer_entropy_exponent_bound} by setting $\tau(n)\in \Omega(n)$. 
    Then this gives us $\tr(\Pi^{(2\alpha)}\psi^{\otimes 2\alpha}) = 2^{(1-\alpha)\Omega(n)} + 2\eta(n) \in o((\poly n)^{-1})$ since $\eta(n) \in \negl_\poly(n) = o((\poly n)^{-1})$ and the $2^{(1-\alpha)\Omega(n)}$ term is dominated by $o((\poly n)^{-1})$.}
    
    By applying $\log$ to both sides of the preceding inequality and dividing both sides by $-(\alpha-1)$, we obtain a lower bound the expected stabilizer R\'enyi entropy of $\psi$ as
    \begin{equation}
    \begin{aligned}
        \E_\psi[M_\alpha(\psi)] &\geq -\frac{\log(2^{-(\alpha-1)\tau(n)} + 2\eta(n))}{\alpha-1} \\
        &= -\frac{\log(2\eta(n)) + \log\Big( 1 + \frac{2^{-(\alpha-1)\tau(n)}}{(2\eta(n))} \Big)}{\alpha-1} \\
        &= -\frac{\log(\negl_\mathbf{T}(n)) + \log\Big( 1 + \frac{2^{-(\alpha-1)\tau(n)}}{\negl_\mathbf{T}(n)} \Big)}{\alpha-1} \\
        &\geq -\frac{\log(\negl_\mathbf{T}(n)) + \frac{2^{-(\alpha-1)\tau(n)}}{\negl_\mathbf{T}(n)} }{\alpha-1} \;,
    \end{aligned}
    \end{equation}
    since $\log(a+b) = \log a + \log(1+b/a)$ and $\log(1+a)\leq a$ for all $0<a<1$ and since $2\eta$ is a $\mathbf{T}$-negligible function.
    Thus, we have shown the first statement of Proposition~\ref{prop:stabilizer_entropy_gap}.

    To show the stabilizer R\'enyi entropy gap in Proposition~\ref{prop:stabilizer_entropy_gap}, we take $\{|\varphi\>\}$ to be the Haar-random ensemble and $\{|\psi\>\}$ to be a $\mathbf{T}$-PRS.
    The expected stabilizer R\'enyi entropy of the Haar random state is given by~\cite[Lemma S2]{Gu_2024},\cite{leone2023clifford}
    \begin{equation}\label{eqn:stab_entropy_Haar}
        \E_\varphi[M_\alpha(\varphi)]=
        \begin{cases}
            n-2 + O(2^{-n}) , \quad&\text{for }\alpha=2 \\
            \frac{n}{\alpha-1} + O(2^{-n}) , \quad&\text{for }\alpha\geq3
        \end{cases} \;.
    \end{equation}
    Hence we can set $\E_\varphi[M_\alpha(\varphi)] = f(n)\in\Theta(n)$ and $\tau\in O(n)$ to obtain
    \begin{equation}
    \begin{aligned}
        \Delta_{M_\alpha}(\{\varphi\},\{\psi\}) &= \Big|\E_\varphi[M_\alpha(\varphi)] - \E_\psi[M_\alpha(\psi)]\Big| \\
        &\leq f(n) + \frac{\log(\negl_\mathbf{T}(n)) + \frac{2^{-(\alpha-1)\tau(n)}}{\negl_\mathbf{T}(n)} }{\alpha-1}
    \end{aligned}
    \end{equation}
    which concludes the proof.
\end{proof}

\begin{table*}[t]
    \centering
    \[
    \begin{array}{|c|c|c|}
        \hline
        \mathbf{T}(n) & \E_\psi[M_\alpha(\psi)] & \Delta_{M_\alpha}(\{\varphi\},\{\psi\}) \leq \\[0.5ex]
        \hline
        O(\poly n) & \frac{\omega(\log n)}{\alpha-1} & O(n) - \frac{\omega(\log n)}{\alpha-1} \\[0.5ex]
        O(n\log n) & \frac{\omega(1) + \log (n\log n)}{\alpha-1} & O(n) - \frac{\omega(1) + \log (n\log n)}{\alpha-1} \\[0.5ex]
        O(n) & \frac{\omega(1) + \log n}{\alpha-1} & O(n) - \frac{\omega(1) + \log n}{\alpha-1} \\[0.5ex]
        O(\poly\log n) & \frac{\omega(\log\log n)}{\alpha-1} & O(n) - \frac{\omega(\log\log n)}{\alpha-1} \\[0.5ex]
        O(\log n) & \frac{\omega(1) + \log\log n}{\alpha-1} & O(n) - \frac{\omega(1) + \log\log n}{\alpha-1} \\[0.5ex]
        \hline
    \end{array}
    \]
    \caption{
    Left column: Computational power $\mathbf{T}$ of the observer. 
    Center column: Expected stabilizer $\alpha$-R\'enyi entropy (for odd integer $\alpha>2$) of ensemble $\{|\psi\>\}$ that is $\mathbf{T}$-indistinguishable from $\{|\varphi\>\}$.
    We assume that
    Right column: Upper bound of stabilizer $\alpha$-R\'enyi entropy gap $\Delta_{M_\alpha}$ between Haar-random ensemble $\{|\varphi\>\}$ and ensemble $\{|\psi\>\}$.
    The quantities in the center and right columns can be obtained directly from Proposition~\ref{prop:stabilizer_entropy_gap} by setting $\tau(n) \in O(n)$ and setting $\mathbf{T}(n)$ as in the left column.
    }
    \label{tab:observer_computational_power_magic_gap}
\end{table*}

First, recall the stabilier R\'enyi entropy of the Haar-random ensemble $\{\varphi\}$ in eqn.~\eqref{eqn:stab_entropy_Haar}, $\E_\varphi[M_\alpha(\varphi)] \in O(n)$.
Thus if we set $\{|\varphi\>\}$ to be the Haar-random ensemble we can set $\tau(n)\in O(n)$.
So for any $\mathbf{T}(n)$ that grows shower than polynomials, it holds that 
\begin{equation}
\begin{aligned}
    \frac{2^{-(\alpha-1)\tau(n)}}{\negl_\mathbf{T}(n)} \in O(2^{-n}) \;.
\end{aligned}
\end{equation}

Now, similar to what we have done for relative entropy of coherence and entanglement entropy, we give a lower bound for expected stabilizer R\'enyi-$\alpha$ entropy of $n$-qubit ensemble $\{|\psi\>\}$ with low-magic that is $\mathbf{T}$-indistinguishable from $n$-qubit Haar-random ensemble $\{|\varphi\>\}$ for different $\mathbf{T}$ along with the magic gap between Haar-random ensemble and $\mathbf{T}$-PRS.
The results are summarized in Table~\ref{tab:observer_computational_power_magic_gap}.
\begin{enumerate}
    \item For $\poly$-time observers ($\mathbf{T}(n)=\poly n$), we have $\eta\in\negl_\poly$, i.e. $\eta(n)<2^{-\omega(\log n)}$ and $\alpha=t(n)\in O(\poly n)$, hence
    \begin{equation}
    \begin{aligned}
        &-\frac{\log(\negl_\mathbf{T}(n)) + \frac{2^{-(\alpha-1)\tau(n)}}{\negl_\mathbf{T}(n)} }{\alpha-1} > \frac{\omega(\log n)}{\alpha-1}
    \end{aligned}
    \end{equation}
    since $\frac{2^{-(\alpha-1)\tau(n)}}{\negl_\mathbf{T}(n)} \in O(2^{-n})$.
    
    \item For linearithmic-time observers, note that $f\in O(n\log n)$ means that for all $c>0$ there exists $N\in\N$ such that $n\geq N \rightarrow f(n) > cn\log n$.
    Thus for such function $f$, it holds that for all $c>0$ there exists $N\in\N$ such that $n\geq N \Rightarrow \log f(n) > \log c n\log n = \log c + \log(n\log n)$.
    In other words, $\log f(n) > \omega(1) + \log(n \log n)$.
    Thus, since we have $\eta\in\negl_{O(n\log n)}$ which implies that $\eta(n)<1/\omega(n\log n)$, we obtain
    \begin{equation}
    \begin{aligned}
        & -\frac{\log(\negl_{n\log n}(n)) + \frac{2^{-(\alpha-1)\tau(n)}}{\negl_{n\log n}(n)} }{\alpha-1} \\
        &= \frac{-\log \eta(n) - O(2^{-n})}{\alpha-1} \\
        &> \frac{\log \omega(n\log n)}{\alpha-1} \\
        &= \frac{\omega(1) + \log(n\log n)}{\alpha-1} \;.
    \end{aligned}
    \end{equation}

    \item For linear-time observers we have 
    \begin{equation}
    \begin{aligned}
        & -\frac{\log(\negl_\mathbf{T}(n)) + \frac{2^{-(\alpha-1)\tau(n)}}{\negl_\mathbf{T}(n)} }{\alpha-1} \\
        &= \frac{-\log \eta(n) - O(2^{-n})}{\alpha-1} \\
        &> \frac{\log \omega(n)}{\alpha-1} \\
        &= \frac{\omega(1) + \log n}{\alpha-1} \;.
    \end{aligned}
    \end{equation}

    \item For polylogarithmic-time observers we have
    \begin{equation}
    \begin{aligned}
        & -\frac{\log(\negl_\mathbf{T}(n)) + \frac{2^{-(\alpha-1)\tau(n)}}{\negl_\mathbf{T}(n)} }{\alpha-1} \\
        &= \frac{-\log \eta(n) - O(2^{-n})}{\alpha-1} \\
        &> \frac{\log 2^{-\omega(\log\log n)}}{\alpha-1} \\
        &= \frac{\omega(1) + \log\log n}{\alpha-1} \;.
    \end{aligned}
    \end{equation}

    \item For logarithmic-time observers we have
    \begin{equation}
    \begin{aligned}
        & -\frac{\log(\negl_\mathbf{T}(n)) + \frac{2^{-(\alpha-1)\tau(n)}}{\negl_\mathbf{T}(n)} }{\alpha-1} \\
        &= \frac{-\log \eta(n) - O(2^{-n})}{\alpha-1} \\
        &> \frac{\log \omega(\log n)}{\alpha-1} \\
        &= \frac{\omega(1) + \log\log n}{\alpha-1}\;.
    \end{aligned}
    \end{equation}
\end{enumerate}

\section{Discussion}

In this work, we extend the notion of pseudorandomness for quantum states from the regime of polynomial-time quantum computers to smaller sized quantum computers.
We propose a framework to construct $\mathbf{T}$-pseudorandom states ($\mathbf{T}$-PRS), a PRS that is computationally indistinguishable from Haar-random states to observers with quantum algorithms which runtime is bounded by a class of functions $\mathbf{T}$.
We derive criteria of such PRS for different classes of functions $\mathbf{T}$ that scales slower than polynomials and give explicit constructions.
Then we define the notion of $\mathbf{T}$-pseudorandom pair, which is a pair of quantum state ensembles possessing different amount of quantum resource, but are indistinguishable to observers with quantum algorithms which runtime bounded by $\mathbf{T}$.
For particular classes of functions $\mathbf{T}(n)$: linearithmic $O(n\log n)$, linear $O(n)$, polylogarithmic $O(\poly\log n)$, and logarithmic $O(\log n)$, we show that the necessary amount of quantum resources (coherence, entanglement, and magic) that the low-resource ensemble must have decreases with $\mathbf{T}(n)$.
As one can construct such a pair with $\mathbf{T}$-PRS and Haar-random ensemble, we further show how the gap between the Haar-random ensemble's resource and the $\mathbf{T}$-PRS's resource increases as $\mathbf{T}(n)$ decreases.
This demonstrated how $\mathbf{T}$-PRS can mimic high-resource states using lesser resource for computationally weaker observers .

Such parameterization with respect to some class of function $\mathbf{T}$ that bounds the computational power of the observer could in principle be extended to other quantum pseudorandom objects, such as pseudorandom density matrices~\cite{bansal2024pseudorandom}, pseudorandom function-like states~\cite{ananth2022pseudorandom,ananth2022cryptography}, pseudorandom unitaries~\cite{ji2018pseudorandom,haug2023pseudorandom,schuster2024random,ma2024construct}, and pseudorandom isometries~\cite{ananth2024pseudorandom}.
Such $\mathbf{T}$-pseudorandom density matrices, $\mathbf{T}$-pseudorandom function-like states, $\mathbf{T}$-pseudorandom unitaries, and $\mathbf{T}$-pseudorandom isometries can be constructed using our framework in Section~\ref{sec:pseudorandomness_and_indistinguishability} by (1) characterizing how many copies that the observer are allowed to have and (2) specifying the negligible probability of the observer distinguishing them from their respective truly random object.

On the other hand, interesting questions can be asked about pseudoresources and a full-fledged computational resource theory.
The field of resource theory~\cite{chitambar2019quantum,gour2024resources} study how quantum resources such as coherence, entanglement, and magic can be characterized and manipulated.
However, how much computational resource is required to prepare states and perform quantum operations have largely been left out of the picture.
We have shown in Section~\ref{sec:T-pseudoresources} that \textit{perceived} quantum resource is relative to how much computational resource the observer has access to.
It is interesting to explore on how one can formulate a \textit{computational resource theory}, where quantification of a quantum resource is relative to the computational power of the observer and where states and operations are further characterized by their computational complexity.
In such resource theory, which states and operations are considered as resourceful is relative to some computationally bounded observer.
Thus one can characterize the \textit{effective} amount of resource that a quantum state has relative to this observer. 
A recent work in this direction has been done for entanglement~\cite{arnon2023computational}, it would be interesting to see how an extension to other quantum resources and to a full computational resource theory where resourceful states and operations are characterized computationally can be made.

\section*{Acknowledgements}

AT is supported by the CQT PhD scholarship and the Google PhD fellowship. KB acknowledges support from Q.InC Strategic Research and Translational Thrust.

\bibliographystyle{unsrt}
\bibliography{references}

\end{document}